\documentclass[times]{stvrauth}

\usepackage{url}

\usepackage{amsmath}
\usepackage{amssymb}
\usepackage{forest}

\usepackage{algpseudocode}
\makeatletter
\let\OldStatex\Statex
\renewcommand{\Statex}[1][3]{%
    \setlength\@tempdima{\algorithmicindent}%
    \OldStatex\hskip\dimexpr#1\@tempdima\relax}

\usepackage{varwidth}
\usepackage{graphicx}

\usepackage{tikz}
\usetikzlibrary{shadows,shapes.geometric,positioning,arrows,shapes.misc,shadings,arrows,automata,positioning, patterns, decorations.pathreplacing, decorations.pathmorphing}

\usepackage{xcolor}

\newcommand{\Pref}{\text{\it Pref}}

\newcommand{\after}{\ensuremath{\text{-\underline{after}-}}}

\newcommand\strongred{\ensuremath{\preceq_{sr}}}

\newcommand\definedInputs{\ensuremath{\Delta}}
\newcommand\faultDomain{\ensuremath{\mathcal{F}(m)}}
\newcommand{\lang}{\ensuremath{L}}
\newcommand{\ii}[1]{\underline{#1}}
\newtheorem{definition}{Definition}
\newtheorem{theorem}{Theorem}
\newtheorem{lemma}{Lemma}

\newcommand{\crd}{\text{CR}}

\begin{document}
    
\runningheads{R. Sachtleben AND J. Peleska}{Effective grey-box testing with partial FSM models}

\title{Effective grey-box testing with partial FSM models}

\author{Robert Sachtleben\affil{1}\textsuperscript{,}\corrauth\ and
    Jan Peleska\affil{1}}

\address{\affilnum{1}University of Bremen, Bremen, Germany}

\corraddr{Robert Sachtleben, University of Bremen, Bremen, Germany. E-mail: rob\_sac@uni-bremen.de}

\cgs{Jan Peleska is partially funded by the Deutsche Forschungsgemeinschaft (DFG, German Research Foundation) -- project number  407708394.}

\begin{abstract}
For partial, nondeterministic, finite state machines, a new  conformance relation called strong reduction is presented. It complements other existing conformance relations in the sense that the new relation is well-suited for model-based testing of systems whose inputs are enabled or disabled, depending on the actual system state. Examples of such systems are graphical user interfaces and systems with   interfaces that can be enabled or disabled in a mechanical way. We present a new test generation algorithm producing complete test suites for strong reduction. The suites are executed according to the grey-box testing paradigm: it is assumed that the state-dependent sets of enabled inputs can be identified during test execution, while the implementation states remain hidden, as in black-box testing. 
It is shown that this grey-box information is exploited by the generation algorithm in such a way that the resulting best-case test suite size is only linear in the state space size of the reference model. Moreover,
examples show that this may lead to significant reductions of test suite size in comparison to true black-box testing for strong reduction. 
\end{abstract}

\keywords{model-based testing; partial finite state machines; complete test suites; conformance testing; grey-box testing}
\maketitle

\section{Introduction}\label{sec:introduction}

In
 this article, a new conformance relation for model-based testing against partial, nondeterministic  FSM models is presented. This relation is called \emph{strong reduction} and complements the well-known \emph{quasi-reduction}~\cite{DBLP:journals/scp/Hierons19,DBLP:conf/fates/PetrenkoY05,DBLP:journals/tse/Hierons17} in a way that, to our best knowledge, has been missing until today: recall that quasi-reduction allows implementations to realize {\it arbitrary} behaviours for inputs that are not   specified in a   state of the partial FSM reference model, after having run through a given IO trace.  Only for inputs that are specified in such a state   of the reference model, implementations are required to show a subset of the behaviours allowed according to the reference model.
Therefore, this conformance relation is best suited for testing against reference models that are incomplete due to lack of information about the expected behaviour in certain situations, or where a separate reference model is used to cover the cases that have not been handled in the first model. 

In contrast to this, the strong reduction conformance relation presented here requires that, after having run through an IO-trace which must also be in the language of the reference model, the implementation always accepts {\it exactly the same inputs} as the reference model and exhibits a subset of behaviours allowed according to the model. This kind of model is suitable when dealing with systems where the inputs are enabled or disabled in a state-dependent way.
Examples for this kind of systems are
\begin{itemize}
\item Graphical user interfaces: The   buttons accessible for mouse clicks, the text fields accessible for keyboard  input, and other typical input widgets (sliders, pull-down menus, \dots) may change, depending on the state of the interface. If an input widget is not visible in a certain interface state, there is no chance to access it from the outside world.

\item  Mechanical interfaces: A typical credit card reader slot, for example, can only accommodate one card. After that, the input of another card is disabled for mechanical reasons, until the current one has been ejected.
\end{itemize}

It should be noted that strong reduction cannot simply be replaced by the well-known ``standard''-notion of reduction which  only requires inclusion of the implementation's IO-language in the  language of the reference model: reduction would allow for a given partial, nondeterministic FSM model that an implementation   disables certain inputs that are enabled in the reference model.

We are interested in automatically generating test suites  from partial, nondeterministic  FSM models that are \emph{$m$-complete} in the sense that every   implementation which is a strong reduction of the reference model will pass such a suite, but every implementation violating strong reduction conformance will fail at least one test case of the suite, provided that the implementation has no more than $m$ states, while the reference model has $n \le m$ states. This type of questions has been investigated for the known   conformance relations language equivalence, reduction, quasi-equivalence, and quasi-reduction by many authors, a survey is given in Section~\ref{sec:related}. 

Apart from being of high interest for the theory of model-based testing, complete test suites are of particular importance in the field of test safety-critical systems, where the test strength of a suite needs to be justified. Through additional techniques such as using input equivalence classes,   complete test suites can be reduced to a manageable size, while still preserving their completeness properties, so that they are practically applicable to embedded control systems of medium complexity~\cite{Huebner2017}.

When investigating complete   suites for testing strong reduction conformance, a \emph{grey-box testing} approach is promising: from the examples listed above, we see that, while the internal state of an implementation still remains hidden as in black-box testing, the inputs enabled in the current implementation state may be revealed. In the case of software testing  graphical user interfaces, for example, the enabled input events can be captured by checking the visibility-status\footnote{In Java, for example, every widget derived from AWT class {\tt Component} has a Boolean getter method {\tt isVisible()}.} of each widget.  When testing implementations with mechanical interfaces, the enabled interfaces can often be identified by visible inspection which can also be automated using image evaluation techniques. Therefore, it is an interesting research question whether the availability of state-dependent information about enabled/disabled inputs will help to reduce the number of test cases to be performed for achieving completeness in a significant way.


The work presented in this article complements results published  in~\cite{hierons_testing_2004,DBLP:journals/scp/Hierons19}, where complete testing theories for quasi-conformance have been presented.
In particular, we consider the following results as the main contributions of this article.
\begin{enumerate}
\item The strong reduction conformance relation is introduced, to our best knowledge, for the first time.
It complements the known conformance relations language equivalence, reduction, quasi-equivalence, and quasi-reduction by providing a suitable means to specify conformance to incomplete models, where state-dependent, unspecified inputs are considered as {\it disabled} in the respective state. 

\item We introduce a new $m$-complete test case generation 
algorithm for checking strong reduction conformance in a grey-box setting, where the true state of the implementation is hidden as in black-box testing, but the state-dependent enabled or disabled inputs can be evaluated during test execution. The algorithm is a new variant of the known \emph{state counting method}~\cite{petrenko_testing_1996,hierons_testing_2004}, with a refined view on reliably distinguishable states and a
further test case reduction adopted from the H-Method~\cite{DBLP:conf/forte/DorofeevaEY05} used in testing for (quasi-)equivalence between FSMs. 

\item
The decrease of test suite size that can be achieved by  exploiting grey-box information is shown by means of complexity calculations and through concrete examples. It is explained why   grey-box testing is particularly advantageous when testing for strong reduction, whereas it could not be applied in a strictly analogous way for quasi-reduction testing.


\end{enumerate}

The algorithm for generating complete   suites for testing strong  reduction conformance has been implemented in the open source library {\it fsmlib-cpp}, as described in Appendix~\ref{appendix:mex_test_suite}.

\subsection{Overview}

In Section~\ref{sec:notation}, the basic notation and well-known facts about finites state machines are summarised, as far as needed for the results presented in this article.  In Section~\ref{sec:crexample}, an example is introduced which serves to motivate that one more conformance relation is needed in addition to the known ones. Moreover, this example is used to calculate a test suite of non-trivial size, using the new test generation strategy presented here. This test suite is too large to be be shown verbatim in this paper, it is available for download under \url{http://www.mbt-benchmarks.org}. In Section~\ref{sec:ssr}, the strong reduction conformance relation presented in this article is formally defined and compared to the existing conformance relations language equivalence, reduction, and quasi reduction.  

In Section~\ref{sec:statecounting}, test cases, pass relation, completeness and test oracles are specified for the purpose of grey-box testing against the strong reduction conformance relation.
In the remainder of this section, the test generation algorithm is presented. To prepare this, an extended notion of deterministically reachable states is defined in Section~\ref{sec:dreach}. Next,  the concept of reliable distinguishability is extended for the purpose of strong reduction testing in Section~\ref{sec:rdis}. An algorithm for computing sets of pairwise reliably distinguishable states is described. Then, in Section~\ref{sec:gen}, the main algorithm for generating complete strong reduction conformance test suites is presented, and its correctness is proven. 
To understand the ``mechanics'' of the test generation algorithm, a test suite derived from a small reference FSM is shown in Appendix~\ref{appendix:mex_test_suite}. There, it is also explained how to use the C++ library {\it fsmlib-cpp} for automatically creating complete   suites for testing strong reduction conformance.
In Section~\ref{sec:complexitymain}, bounds for corner cases of the test suite size are discussed. It is shown, how the grey-box evaluation of  enabled and disabled inputs can result in significant test suite reductions. The detailed  proofs of the test suite size bounds are presented in Appendix~\ref{sec:complexityproofs}. 
In Section~\ref{sec:relationrelations}, it is explained why complete test generation algorithms for reduction, quasi-reduction, and strong reduction differ significantly. In Section~\ref{sec:r0impact}, the subtle differences in r-distinguishability to be observed when comparing algorithms for quasi-reduction and strong reduction are discussed, and it is shown by means of examples, how these differences affect the test suite size.

In Section~\ref{sec:related} we discuss related work.  The conclusion is presented in Section~\ref{sec:comc}.


\section{Notation and Background}\label{sec:notation}

In this section, we introduce notation, definitions, and basic facts about finite state machines, as used in this article and related  publications, such as~\cite{DBLP:journals/scp/Hierons19,DBLP:conf/fates/PetrenkoY05,DBLP:journals/tse/Hierons17}.

In model-based testing (MBT), a 
\emph{system under  test (SUT)} is verified  by means of a systematic application of inputs to the SUT, where for each applied input the observed response is compared to the behaviours allowed by the reference model.
Depending on the underlying MBT test case generation strategy, inputs ``of interest'' are also identified using some kind of  model analysis.
 Here we assume that the SUT can be \textit{reset} to its initial state at any time, for example by switching it off and then on again.
For convenience, we write sequences of input-output (IO) pairs $(x_1,y_1).\ldots.(x_n,y_n)$ as $x_1\ldots x_n/y_1 \ldots y_n$ and use $\bar{x}/\bar{y}$ to denote sequences with \textit{input portion} $\bar{x}$ and \textit{output portion} $\bar{y}$. Furthermore, we also use $\alpha, \beta, \pi$, and $\tau$ to denote IO sequences, while $\epsilon$ denotes the empty sequence. Concatenation of sequences $\alpha$ and $\beta$ is denoted by $\alpha.\beta$.
For any sequence $\alpha$, let $\Pref(\alpha) = \{ \alpha_1~|~\exists \alpha_2 : \alpha=\alpha_1.\alpha_2 \}$ denote the set of prefixes of $\alpha$. We say that $\alpha_1$ is a \textit{proper} prefix of $\alpha$ if it is a prefix of $\alpha$ and also shorter than $\alpha$. Function $\Pref$ can be lifted to sets of sequences such that $\Pref(A) = \bigcup_{\alpha \in A} \Pref(\alpha)$. Finally, for sets of sequences $A$ and $B$ we use $A.B$ to denote the extension of every sequence in $A$ with every sequence in $B$ and define $A.\varnothing$ to result in $A$. 

A \textit{finite state machine (FSM)} $M= (S,\ii s,\Sigma_I,\Sigma_O,h_M)$ is a 5-tuple consisting of a finite set $S$ of \emph{states}, an \emph{initial state} $\ii s \in S$, finite sets $\Sigma_I$ and $\Sigma_O$ constituting the \emph{input and output alphabet}, respectively, and a \emph{transition relation} $h_M \subseteq S \times \Sigma_I \times \Sigma_O \times S$. The interpretation of the fact   $(s_1,x,y,s_2) \in h_M$ is that 
there exists  a transition in $M$ from $s_1$ to $s_2$ for input $x$ that produces output $y$. For $s \in S$ and $x \in \Sigma_I$ we write $out(s,x)$ to denote the set $\{y~|~\exists s' : (s,x,y,s') \in h_M\}$ of all possible outputs produced by $s$ in response to $x$.
We define the \emph{size} of $M$, denoted by $|M|$, as the number $|S|$ of states it contains.
The \textit{language} $\lang_M(s_0)$ of some state $s_0$ of $M$ denotes the set of all sequences $\bar{x}/\bar{y} \in (\Sigma_I \times \Sigma_O)^*$ of IO pairs such that $M$ can react to $\bar{x}$ applied to $s_0$ with outputs 
$\bar{y}$. More formally, if $\bar x = x_1\dots x_k$ and $\bar y =  y_1\dots y_k$, then $\bar x / \bar y\in L_M(s_0)$, if and only if, there exist states $s_1,\dots s_k\in S$, such that
\begin{equation}\label{eq:path}
  \forall i = 1,\dots,k : (s_{i-1},x_i,y_i,s_i) \in h_M
\end{equation}
The language of $M$ itself, denoted $\lang({M})$, is the language of its initial state, that is, 
$L(M) = L_M(\ii s)$. 

In some situations, we are interested in all input sequences $\bar x$ of a given length $|\bar x| = i$. To this end, the notation $\Sigma_I^i$ is used to denote the set of all possible input sequences of length $i$ over alphabet $\Sigma_I$. By convention, $\Sigma_I^0$ denotes the set $\{ \epsilon \}$ containing only the empty sequence.

We assume that all states $s$ of $M$ are \emph{reachable}. This means that for any $s\in S-\{\ii s\}$, there always exists $k > 0$, $s_1,\dots s_{k-1}$, $\bar x = x_1\dots x_k$ and $\bar y =  y_1\dots y_k$, such that 
formula~\eqref{eq:path} holds with $s_0 = \ii s$ and $s_k = s$. The initial state $\ii s$ is reached by the empty sequence $\epsilon$. If an initial definition of $M$ contains unreachable states, these can be identified by means a breadth-first-search starting at the initial state and visiting the direct successors of each state linked via $h_M$.

An FSM $M$ is called \emph{observable}, if for each  state $s$, input $x$ and output $y$ there exists at most one state $s'$ in $S$ such that $(s,x,y,s') \in h_M$. That is, the target state reached from some state with some input can be uniquely determined using the observed output. This property also extends to IO-sequences, as the state reached by an IO-sequence $\alpha \in \lang_M(s)$ applied to state $s$, denoted by $s\after \alpha$, is again uniquely determined. In the remainder of this paper, we assume every FSM to be observable, since there exist algorithms transforming a  non-observable FSM into  an observable one with the same 
language~\cite[Appendix~II]{luo_test_1994}.

An input $x$ is \textit{defined} in state $s$ if $out(s,x) \neq \varnothing$, and the set of all inputs defined in $s$ for $M$ is denoted $\definedInputs_M(s)$. 
An FSM is called \emph{completely specified} if and only if, 
$\definedInputs_M(s) = \Sigma_I$ for all states $s\in S$. 
Equivalently, this means that   $|out(s,x)| > 0$ for all $s\in S$ and $x\in\Sigma_I$. 
An FSM which is {\it not} completely specified is called \emph{partial}.
As we assume any FSM to be observable, we sometimes write $\definedInputs_M(\alpha)$ for an IO sequence $\alpha$, instead of 
$\definedInputs_M(\ii s\after \alpha)$. Note that $\definedInputs_M(\alpha) = \varnothing$ if $\alpha\not\in L(M)$.

A (completely specified or partial) FSM is called \emph{deterministic} if and only if $|out(s,x)| \le 1$ for all $s\in S$ and $x\in\Sigma_I$. This means that each output is uniquely determined by the current state and the 
selected input.

An FSM $I$ is called a \emph{reduction} of another FSM $M$, if both operate on the same input and output alphabets and $\lang(I) \subseteq \lang(M)$ holds. Reduction is a suitable conformance relation for ensuring safety properties: $\lang(I) \subseteq \lang(M)$ asserts that the implementation $I$ will never produce an IO sequence which is not in the language of $M$. Therefore, if the reference model $M$ is considered as safe, $I$ will be safe as well.

Reduction is usually considered in the context of nondeterministic reference 
models~\cite{petrenko_testing_2011,DBLP:conf/hase/PetrenkoY14,hierons_testing_2004} or if incomplete implementations~\cite{DBLP:journals/scp/Hierons19,DBLP:conf/fates/PetrenkoY05} are allowed. For deterministic, completely specified reference models $M$ it is easy to see that any completely specified  reduction $I$ of $M$ must already be language-equivalent to $M$: Since $I$ is complete, it has to accept any input sequence $\bar x \in \Sigma_I^*$. Since $M$ is deterministic, there is exactly one output sequence $\bar y$ fulfilling $\bar x/\bar y\in \lang(M)$. Because $\lang(I)\subset \lang(M)$ holds,
$\bar y$ is the {\it only} output sequence $I$ is allowed to produce  in reaction to $\bar x$. 
This proves
$\lang(I) = \lang(M)$. 

For partial nondeterministic FSMs, an additional conformance relation has been proposed in~\cite{DBLP:journals/scp/Hierons19,DBLP:conf/fates/PetrenkoY05,DBLP:journals/tse/Hierons17} which takes partiality into account\footnote{The definitions given in~\cite{DBLP:journals/scp/Hierons19,DBLP:conf/fates/PetrenkoY05,DBLP:journals/tse/Hierons17} slightly differ; we use here a definition which is equivalent to~\cite[Definition~5]{DBLP:journals/scp/Hierons19}.}: FSM $I$ with initial state $\ii u$ is a \emph{quasi-reduction} of $M$ with initial state $\ii s$ if and only if the following properties hold for
all $\alpha\in \lang(I)\cap \lang(M)$ and $x \in\definedInputs_M(\alpha)$.
\begin{align}
& x \in \definedInputs_I(\alpha) \\
&\{ y\in\Sigma_O~|~\alpha.(x/y)\in \lang(I) \} \subseteq \{ y\in\Sigma_O~|~\alpha.(x/y)\in \lang(M) \}
\end{align}
Quasi-reduction implies that the implementation $I$, after having run through an IO sequence $\alpha$ which is also contained in the language of the reference model,
 will accept {\it at least} the inputs   accepted by the reference model after $\alpha$. For all
inputs  $x$ accepted by $M$ after $\alpha$, the implementation
 will produce a subset of the outputs possible in $M$. For input sequences {\it not} accepted by $M$, $I$ may produce arbitrary behaviour. Obviously, quasi-reduction implies reduction and completeness of $I$ for the case where $M$ is completely specified. For partial reference FSMs $M$, however, neither quasi-reduction, nor   reduction implies the other.


\section{Motivating Example}
\label{sec:crexample}

In this section, an example of an FSM reference model is given which serves to motivate the new conformance 
relation defined in Section~\ref{sec:ssr} and to generate test cases using to the adapted and optimised 
state counting algorithm described in Section~\ref{sec:statecounting}.

\begin{figure}[ht]
    \begin{center}
        \includegraphics[width=.85\textwidth]{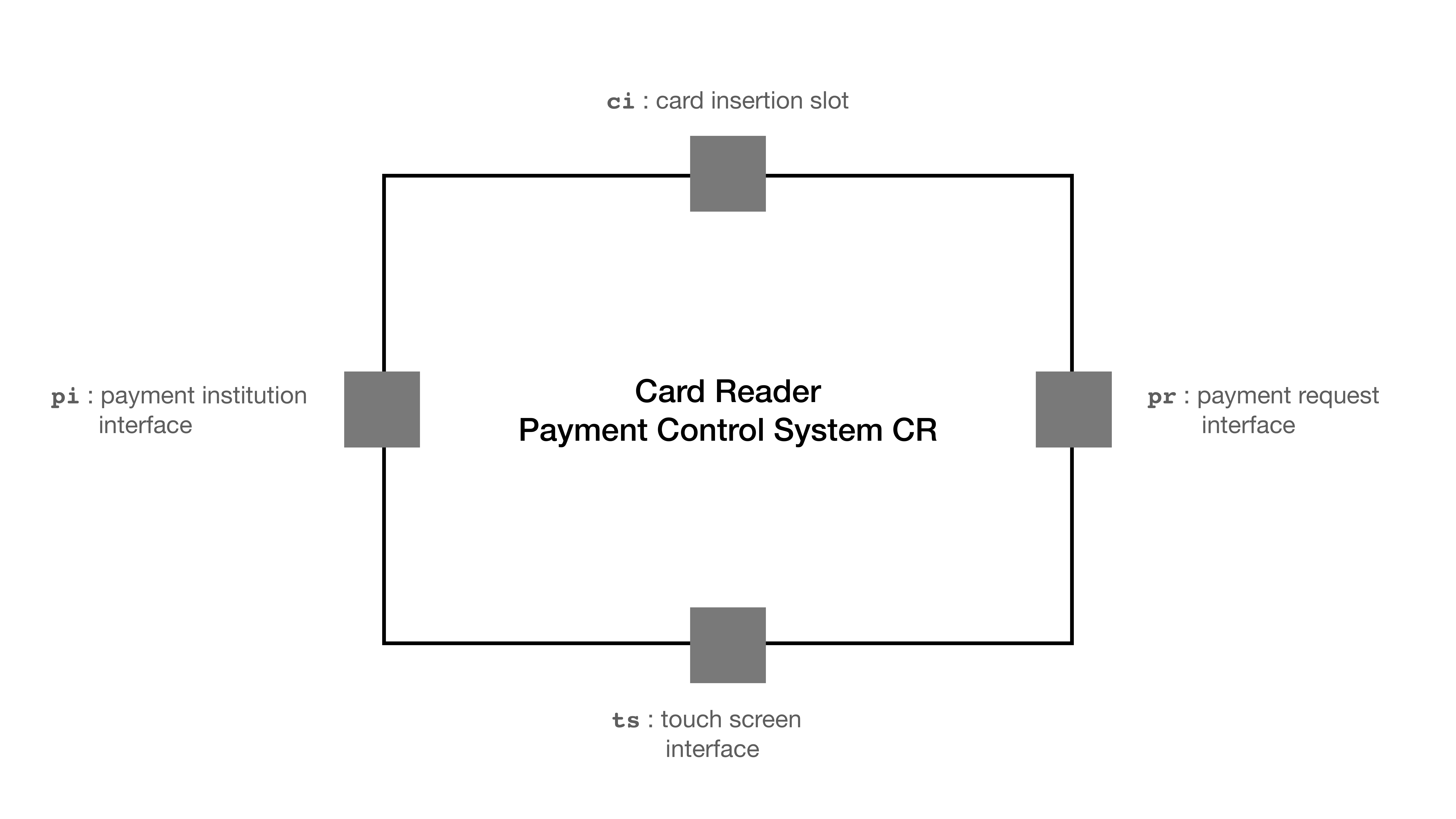}
    \end{center}
    \caption{CR interfaces.}
    \label{fig:crinterfaces}
\end{figure}

The \emph{Card Reader Payment Control System (CR)} is a standard device used, for example, in super markets or ticket vending machines to handle secure authorisation of payment amounts specified by electronic cash registers for vending machines. The CR interfaces   are shown in Fig.~\ref{fig:crinterfaces}. Interface {\tt ci} allows to insert a credit card which is then moved 
into the system so that it cannot be removed until  the transaction has been completed. 
At the end of the transaction, the card is ejected, and a sensor indicates whether it has been removed from the slot by its owner. 
Interface {\tt pr} receives payment requests from cash registers or vending machines. Interface {\tt pi} sends authorisations for the requested funds transfers from card holders' bank accounts to the vendors' accounts. Interface {\tt ts} represents a touch screen interface whose sub-interfaces (output text fields, input keypad, different kinds of buttons) change during the transaction.

\begin{table}[htp]
    \caption{Input alphabet $\Sigma_I^\crd$ of the card reader state machine $\crd$.}
    \begin{center}
        \begin{tabular}{|l|l|}\hline\hline
            {\bf Input} & {\bf Description} \\\hline\hline
            {\tt pr.A} & payment request for a large amount \\\hline
            {\tt pr.a} & payment request for a small amount \\\hline
            {\tt ci.in.v} & valid card insertion into the reader's slot \\\hline
            {\tt ci.in.i} & invalid card insertion into the reader's slot \\\hline
            {\tt ci.r} & removal of an ejected card \\\hline
            {\tt ts.in.ok} & authorise payment command on touch screen\\\hline
            {\tt ts.in.ab} & abort transaction command on touch screen\\\hline
            {\tt ts.in.vp} & entry of a valid PIN via touch screen\\\hline
            {\tt ts.in.ip} & entry of an invalid PIN via touch screen\\
            \hline\hline
        \end{tabular}
    \end{center}
    \label{table:input}
\end{table}%

\begin{table}[htp]
    \caption{Output alphabet $\Sigma_o^\crd$  of the card reader state machine $\crd$.}
    \begin{center}
        \begin{tabular}{|l|l|}\hline\hline
            {\bf Output} & {\bf Description} \\\hline\hline
            {\tt ts.out.ic} & insert-card request on touch screen  \\\hline
            {\tt ts.out.aut} & request to authorise payment amount on touch screen  \\\hline
            {\tt ts.out.p} & request PIN entry on touch screen  \\\hline
            {\tt ts.out.ip} & `invalid PIN' message with request to re-enter PIN on touch screen  \\\hline
            {\tt ts.out.cw} & `Card withdrawn' message on touch screen  \\\hline
            {\tt ts.out.clr} & clear touch screen  \\\hline
            {\tt ci.out} & card is ejected (remains still in the slot)  \\\hline
            {\tt pi.aut} & payment authorisation message   \\
            \hline\hline
        \end{tabular}
    \end{center}
    \label{table:output}
\end{table}%

The formal behavioural specification of the CR is modelled by the finite state machine 
$\crd = (S_\crd,{\tt init},\Sigma_I^\crd, \Sigma_O^\crd,h_\crd)$
specified in Fig.~\ref{fig:crbehavior}. Its input alphabet is specified in Table~\ref{table:input}, and its output alphabet in Table~\ref{table:output}.

While no payment request is present (state {\tt init}), the insertion of valid or invalid credit cards ({\tt ci.in.v}, {\tt ci.in.i}) leads immediately to their ejection ({\tt ci.out}). They have to be removed ({\tt ci.r}) before the system returns to its initial state {\tt init}, where it is ready to accept the next payment request. 

When a payment request arrives ({\tt pr.a}, {\tt pr.A}), users are requested via touch screen output {\tt ts.out.ic} to insert their credit card. The concrete payment amounts that are requested are abstracted in the input alphabet of the FSM   to large amounts ({\tt pr.A}) and small amounts ({\tt pr.a}).

Card insertion is abstracted to FSM inputs {\tt ci.in.v} for insertion of a valid credit card and
{\tt ci.in.i} for an invalid card. Invalid cards are  ejected again from the card insertion slot ({\tt ci.out}). After removal of the card ({\tt ci.r}), the CR resumes its initial state.

After a valid card has been inserted, a request to authorise the payment amount is displayed on the touch screen ({\tt ts.out.aut}). Now it becomes possible to give touch screen commands `authorise payment' ({\tt ts.in.ok}), or
`abort transaction' ({\tt ts.in.ab}).
A transaction abort leads to ejection of the card and return to the initial FSM state, after the card has been removed as described above. 

After authorisation of the amount, the behaviour depends on the payment amount to be authorised. (1) If it is a large amount, the entry of the card's PIN number is requested ({\tt ts.out.p}). After a valid PIN entry 
({\tt ts.in.vp}), the
card is ejected, and an authorisation message ({\tt pi.aut})
is sent to the payment institution, after the card has been removed.
(2) If a small amount has been authorised, a nondeterministic decision is performed: either, 
the card is immediately ejected, and an authorisation message  is sent to the payment institution, after the card has been removed, or the PIN entry is requested just as for large amounts to be paid.

In any case, the PIN number entry is abstracted to inputs `valid PIN' ({\tt ts.in.vp}) and 
`invalid PIN' ({\tt ts.in.ip}) in the CR model. If an invalid PIN is entered,
a second and third input is possible. 
While trying to enter a new PIN, it is always possible to abort the transaction ({\tt ts.in.ab}).
After three invalid inputs, the card is withdrawn by the CR ({\tt ts.out.cw}), and the initial state is resumed.

\begin{figure*}[ht]
    \begin{center}
        \includegraphics[width=\textwidth]{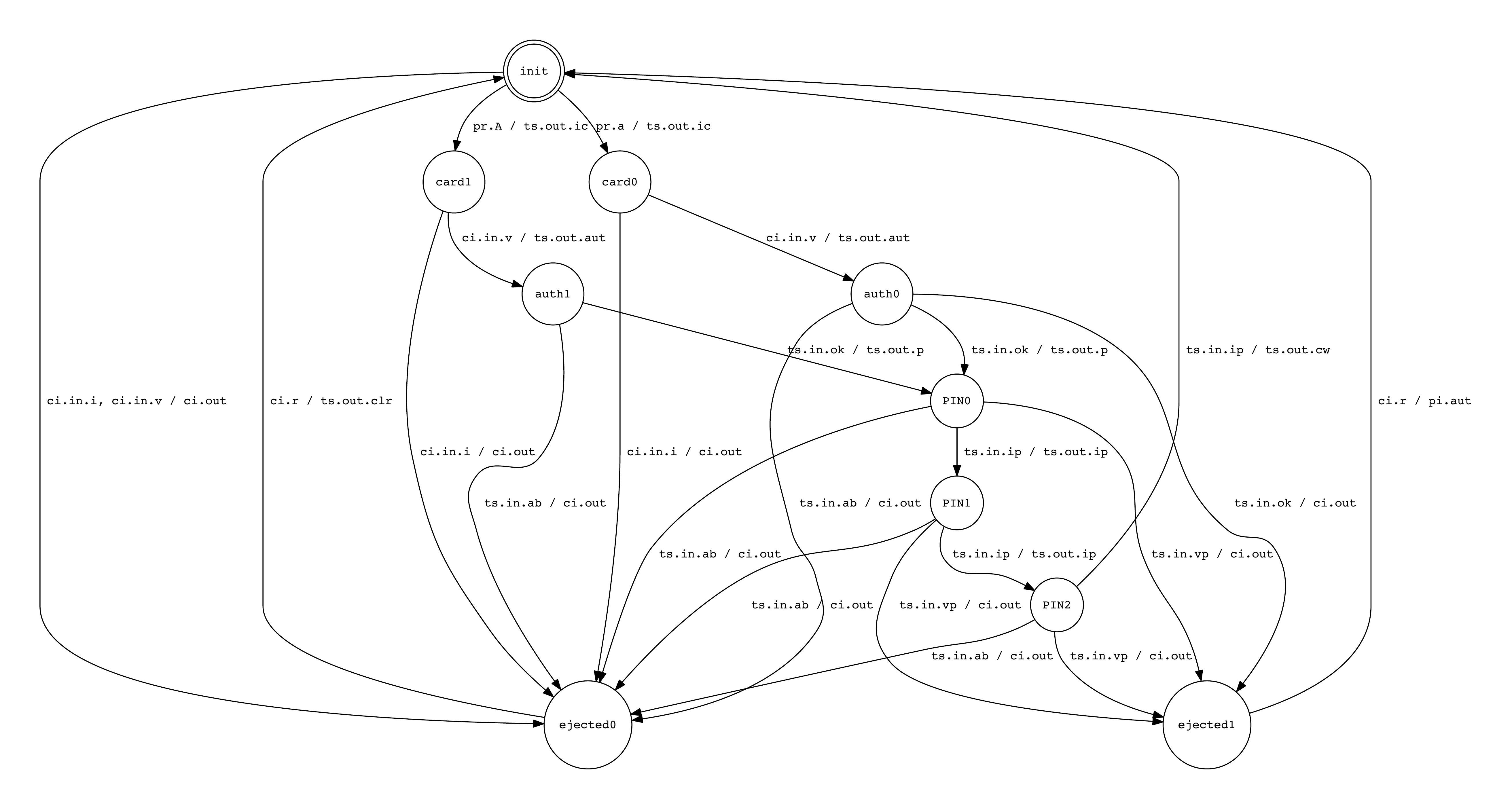}
    \end{center}
    \caption{FSM model $\crd = (S_\crd,{\tt init},\Sigma_I^\crd, \Sigma_O^\crd,h_\crd)$ of the card reader payment control system.}
    \label{fig:crbehavior}
\end{figure*}

As described above and modelled in Fig.~\ref{fig:crbehavior}, the FSM describing the CR behaviour is not completely specified. The unspecified inputs in each state can be separated into two classes.
\begin{itemize}
    \item \textbf{Ignored inputs:} The unspecified input is just an abbreviation for a self-loop transition labelled by this input with an auxiliary member of the output alphabet indicating `no output'. 
    
    \item \textbf{Disabled inputs:} It is impossible to provide this input in the state under consideration.  
\end{itemize}
Table~\ref{table:acc} specifies the input events that are ignored or disabled  in each state. 
For example,   input {\tt ci.r} is disabled in all   states $s$ but {\tt ejected0} and {\tt ejected1}, since in these $s$, either no card is present, or the card has been moved into the system, so that it is 
mechanically impossible to remove it. Similarly, it is impossible to insert a card in any state but {\tt init}, because it is impossible to insert a second card while there is already one present. On the touch screen display, input {\tt ts.in.ok} is only possible in states {\tt auth1} and {\tt auth2}, because the `ok' button is not displayed in the other states and can therefore not be pressed. In contrast to these disabled events,
the inputs on interface {\tt pr} are just ignored in all states but {\tt init}.

\begin{table}[htp]
    \footnotesize
    \caption{State-dependent ignored and disabled inputs  of the CR state machine.}
    \begin{center}
        \begin{tabular}{|l|l|l|}\hline\hline
            {\bf State} & {\bf Ignored inputs} & {\bf Disabled inputs} \\\hline\hline
            {\tt init} & $\varnothing$ & $\{ {\tt ci.r}, {\tt ts.in.ok}, {\tt ts.in.ab},{\tt ts.in.vp}, {\tt ts.in.ip}  \}$ \\\hline
            {\tt card0} & $\{ {\tt pr.a}, {\tt pr.A}\}$ & $\{ {\tt ci.r}, {\tt ts.in.ok}, {\tt ts.in.ab}, {\tt ts.in.vp}, {\tt ts.in.ip} \}$ \\\hline
            {\tt card1} & $\{ {\tt pr.a}, {\tt pr.A}\}$ & $\{ {\tt ci.r}, {\tt ts.in.ok}, {\tt ts.in.ab}, {\tt ts.in.vp}, {\tt ts.in.ip} \}$ \\\hline
            {\tt auth0} & $\{ {\tt pr.a}, {\tt pr.A}\}$ & $\{ {\tt ci.r}, {\tt ci.in.v}, {\tt ci.in.i},   {\tt ts.in.vp}, {\tt ts.in.ip} \}$ \\\hline
            {\tt auth1} & $\{ {\tt pr.a}, {\tt pr.A}\}$ & $\{ {\tt ci.r}, {\tt ci.in.v}, {\tt ci.in.i},   {\tt ts.in.vp}, {\tt ts.in.ip} \}$ \\\hline
            {\tt ejected0} & $\{ {\tt pr.a}, {\tt pr.A}\}$ & $\{   {\tt ci.in.v}, {\tt ci.in.i}, {\tt ts.in.ok}, {\tt ts.in.ab}, {\tt ts.in.vp}, {\tt ts.in.ip} \}$ \\\hline
            {\tt ejected1} & $\{ {\tt pr.a}, {\tt pr.A}\}$ & $\{  {\tt ci.in.v}, {\tt ci.in.i}, {\tt ts.in.ok}, {\tt ts.in.ab}, {\tt ts.in.vp}, {\tt ts.in.ip} \}$ \\\hline
            {\tt PIN0} & $\{ {\tt pr.a}, {\tt pr.A}\}$ & $\{ {\tt ci.r}, {\tt ci.in.v}, {\tt ci.in.i},   {\tt ts.in.ok}  \}$ \\\hline
            {\tt PIN1} & $\{ {\tt pr.a}, {\tt pr.A}\}$ & $\{ {\tt ci.r}, {\tt ci.in.v}, {\tt ci.in.i},   {\tt ts.in.ok}  \}$ \\\hline
            {\tt PIN2} & $\{ {\tt pr.a}, {\tt pr.A}\}$ & $\{ {\tt ci.r}, {\tt ci.in.v}, {\tt ci.in.i},   {\tt ts.in.ok}  \}$ \\
            \hline\hline
        \end{tabular}
    \end{center}
    \label{table:acc}
    \normalsize
\end{table}%

\section{Strong Reduction -- a new Conformance Relation}\label{sec:ssr}

The example presented above  -- though being quite realistic and practical -- shows that none of the established conformance relations are suitable for an implementation $I$ of the reference FSM $\crd$ presented
in Section~\ref{sec:crexample}.

\noindent
(A) Language equivalence is not suitable as a conformance relation: an implementation of the nondeterministic decision whether to require a PIN entry for authorising a small payment will be implemented in a way which guarantees that at least one authorisation request will be made within a limited number of payments with small amounts. The reference model $\crd$ allows for unbounded sequences of small amount payments without PIN request.

\noindent
(B) Reduction is not suitable because it would allow for an implementation corresponding to an empty FSM just
``executing'' the empty IO sequence.

\noindent
(C) Quasi-reduction is not suitable because it would allow implementations that accept inputs disabled by the reference model and exhibit arbitrary behaviour after these inputs. For example,  an implementation could accept payment authorisations in the initial state, without a card having been inserted, and still be a quasi-reduction of $\crd$.

A suitable conformance relation for the example from  Section~\ref{sec:crexample} and -- more general -- for reactive systems involving user interfaces with state-dependent input interfaces should have the following properties.
\begin{enumerate}
\item The implementation shall not exhibit any behaviours disallowed by the reference model.
\item All input sequences allowed by the reference model shall also be allowed by the implementation.
\end{enumerate}

These consideration lead to the following new definition.
\begin{definition}
    FSM $I$ is a \textit{strong reduction} of $M$, denoted $I \strongred M$, if the following holds:
    \begin{align}
        \label{eq:isred}
        &\lang(I) \subseteq \lang(M) \\
        &\wedge \ \forall \alpha \in \lang(I): \definedInputs_I(\alpha) = \definedInputs_M(\alpha)
    \end{align} 
\end{definition}
Note that  Condition~\eqref{eq:isred}   requires that $I$ is a reduction of $M$. 

It should be emphasised that all  the conformance relations discussed above have their specific applications, so there is no ``best'' relation rendering the others superfluous: (A) Language equivalence is typically used when no degrees of freedom should be left for the implementation, that is, when exactly the specified behaviour should be implemented, neither more, nor less. (B) Reduction is typically used to verify that a detailed implementation model still satisfies the safety properties of a reference model. This means that the requirement {\it ``the implementation shall be a reduction of $M$''} is never used as the only requirement for the system to be built, but as a safety-related additional postulate. (C) Quasi-reduction is the conformance relation to be chosen when 
dealing with {\it incomplete} specification models. The absence of an input in a certain state has the meaning {\it ``we do not know what happens here, since this will be determined (later) in another partial reference model''}. (D) Strong reduction is chosen if the reference model is partial because certain inputs cannot happen in certain situations. Typically, this occurs in complex user interfaces, where certain buttons to be pressed only occur in specific system states. Also, as exemplified above, inputs may be impossible due to mechanical reasons.

In the remainder of this paper let FSM $M$ represent the reference model to test against and  assume that the SUT behaves like an unknown member $I$ of the \textit{fault domain} $\faultDomain$, which is the set of all observable  FSMs of size at most equal to some fixed $m \geq |M|$.

%
%
%
%

\section{A Modified State Counting Strategy for Strong Reduction Tests}\label{sec:statecounting}

\subsection{Overview}
In this section, we will elaborate a novel    strategy which is optimised for testing against the strong reduction conformance relation. It will turn out that this -- while still being complete -- leads to significantly fewer test cases, compared to test suites created using conventional reduction testing strategies.   The new strategy is a substantial adaptation of the well-known \emph{state counting} methods investigated, for example, in~\cite{petrenko_testing_1996,hierons_testing_2004}.  
The test strategy described in \cite{hierons_testing_2004} generates test suites by extending a state cover of $M$ by traversal sets, and then extending the resulting sequences by characterisation sets designed to test whether sequences that reach reliably distinguishable states in $M$ also reach distinct states in $I$. We adapt this strategy by modifying the definitions of deterministic reachability and reliable distinguishability in the context of partial FSMs. Moreover, we reduce the test suite extension with 
characterising  sets: this is achieved by checking on-the-fly during test generation, whether the test suite created so far already contains a suitable distinguishing sequence for the actual state pair under consideration. This technique
is an adaptation of the one introduced for the H-Method~\cite{DBLP:conf/forte/DorofeevaEY05} used in testing for (quasi-)equivalence between FSMs.

\subsection{Test Cases and Pass-Relation}
To simplify the presentation, we consider here test suites $T \subseteq \Sigma_I^*$ consisting of input sequences only, but the strategy presented here can easily be adapted to produce so called adaptive test cases (see \cite{DBLP:conf/hase/PetrenkoY14}). We say that $I$ \textit{passes} a test case $\bar{x} \in T$ if for all prefixes $\bar{x}_1 \in \Pref(\bar{x})$ and IO sequences $\bar{x}_1/\bar{y}_1 \in \lang(I)$ it holds that $\bar{x}_1/\bar{y}_1 \in \lang(M)$ and $\definedInputs_I(\bar{x}_1/\bar{y}_1) = \definedInputs_M(\bar{x}_1/\bar{y}_1)$. Note here that if $I$ passes $\bar{x}$, then $I$ also passes all prefixes of $\bar{x}$. We say that $I$ passes $T$ if $I$ passes all $\bar{x} \in T$.

\subsection{Complete Test Suites}

Black-box or grey-box tests do not allow for an inspection of all internal aspects (e.g.~states) of the SUT during a test execution. Therefore, it is not possible to guarantee that a test suite will reveal every conformance violation of every implementation without imposing additional hypotheses. The latter are usually specified by means of a \emph{fault domain} ${\cal F}$ which consists of a set of FSMs that may or may not conform to the reference model~\cite{DBLP:journals/tse/Hierons17}. It is then assumed that the true behaviour of the SUT is equivalent to one FSM model $I$ contained in the fault domain.

A test suite $T$ is \emph{sound} with respect to a given reference model $M$, conformance relation $\le$, and fault domain ${\cal F}$, if  
every SUT whose behaviour is equivalent to some FSM $I\in{\cal F}$ conforming to $M$ (i.e.~$I\le M$) passes every test in $T$.
A test suite $T$ is \emph{exhaustive} with respect to $M$, $\le$, and ${\cal F}$, if  
every SUT whose behaviour is equivalent to some FSM $I\in{\cal F}$ {\it not} conforming to $M$ (i.e.~$I\not\le M$) fails at least one  test in $T$.
A test suite $T$ is \emph{complete} with respect to $M$, $\le$, and ${\cal F}$, if it is both sound and exhaustive.

In this article, we consider the so-called \emph{$m$-completeness} which denotes completeness with respect to
the fault domain ${\cal F}(\Sigma_I,\Sigma_O,m)$ of all state machines over the same input alphabet $\Sigma_I$ and output alphabet $\Sigma_O$ as the reference model, and with at most $m$ states. In the following, will abbreviate ${\cal F}(\Sigma_I,\Sigma_O,m)$ to $\faultDomain$, as $\Sigma_I$ and $\Sigma_O$ are uniquely determined by the reference model.

\subsection{Test Oracles}\label{sec:oracles}

For practical testing, we adopt the usual \emph{fairness assumption} (sometimes called \emph{complete testing assumption})~\cite{hierons_testing_2004}: we assume the existence of a known constant $k\in \mathbf N$, such that a nondeterministic SUT will exhibit {\it every} possible behaviour in response to input sequence $\bar x$, if $\bar x$ is executed at least $k$ times against the SUT. 
Moreover, we assume that a simulation of the reference model $M$ exists which can run through exactly the traces from $L(M)$ and provides the set $\Delta_M(s)$ of inputs accepted in each state $s$.
During   execution of a test case $\bar x$, the simulation is run in back-to-back fashion with the SUT, so that 
the outputs produced by the SUT as reaction to each input can be compared to the outputs possible for the reference model. Moreover, we assume that the tests are \emph{grey-box tests} in the sense that the SUT 
reveals its accepted inputs   $\Delta_I(s')$ in every SUT state $s'$. 

Observe that the grey-box test assumption is quite realistic in many testing scenarios: (1) For software testing of graphical user interfaces, the graphical elements like buttons or input text fields can  be checked by the test harness with respect to visibility. (2) Hardware interfaces like the card insertion slot from our main example in Section~\ref{sec:crexample} directly reveal whether an input is mechanically enabled or disabled.


\subsection{Deterministically Reachable States}\label{sec:dreach}

As $M$ may be both nondeterministic and partial, an input sequence $\bar{x} \in \Sigma_I^*$ may reach between zero and $|M|$ states of $M$. Furthermore, $\bar{x}$ may reach fewer states in $I$ than in $M$ even if $I$ conforms to $M$.
To ensure that each sequence $\bar{x}$ in a state cover reaches exactly one state in $M$, we consider only certain input sequences. We say that an input sequence $\bar{x}$ \textit{deterministically reaches (d-reaches)} a state $s \in S$ if $s$ is the only state reached by $\bar{x}$ in $M$ and $\bar{x}$ is also \textit{strongly defined} in $M$, which requires that for any prefix $\bar{x}_1x$ of $\bar{x}$ and IO sequence $\bar{x}_1/\bar{y}_1 \in \lang(M)$ input $x$ is defined in $\ii s\after \bar{x}_1/\bar{y}_1$. That is, $\bar{x}$ d-reaches $s$ if $\bar{x}$ reaches $\{s\}$ in any strong reduction $M'$ of $M$ that is created by removing transitions from $M$. We say that state $s$ is \textit{d-reachable} if there exists an input sequence that d-reaches $s$.
Note that the initial state of any machine $M$ is always deterministically reachable by the empty input sequence.

A \textit{state cover} $V \subseteq \Sigma_I^*$ of $M$ then is a minimal set of input sequences such that for any d-reachable state $s \in S$, set $V$ contains some $\bar{v}$ that d-reaches $s$. In particular, we require $V$ to contain the empty input sequence $\epsilon$, which d-reaches $\ii s$. To calculate a state cover for a possibly nondeterministic and incomplete FSM $M$ we modify the procedure described in \cite{petrenko_testing_1996}: First we delete the outputs on the transitions of $M$ and complete the result by adding a new state $s_\bot \notin S$, a transition to $s_\bot$ for each $s \in  S$ and $x \in \Sigma_I$ such that $x \notin \definedInputs_M(s)$, and a self-loop on $s_\bot$ for all $x \in \Sigma_I$. Next, we determinise this automaton using standard techniques \cite{hopcroft1979introduction}. Then, state $s$ of $M$ is d-reachable by any input sequence $\bar{x}$ that reaches $\{s\}$ in the determinised automaton and thus we finally create $V$ by selecting for each d-reachable $s \in S$ one such input sequence, in particular selecting $\epsilon$ for $\ii s$.

Consider for example FSM $M_{ex}$ given in Fig.~\ref{fig:syntheticExample} with input alphabet $\{a,b\}$ and output alphabet $\{0,1,2,3\}$. The initial state of this FSM behaves nondeterministically on any given input, but state $s_2$ can still be deterministically reached for example by sequence $a.b$, as $a$ is defined in the initial state, $b$ is defined in any state reached by $a$ and applying $b$ to any such state reaches $s_2$. This can be verified using the technique described above for the calculation of a state cover, which results in the determinised automaton given in Fig.~\ref{fig:syntheticExampleDet}. In this automaton, sequence $a.b$ reaches $\{s_2\}$ and thus d-reaches $s_2$ in $M_{ex}$. The automaton also shows that neither $s_1$ nor $s_3$ can be d-reached, and that sequences such as $a.a.b$ that reach only a single state in $M_{ex}$ are not selected by the technique because 
they are not strongly defined and thus reach states containing $s_\bot$ in the determinised automaton. As a result, $V_{ex} = \{ \epsilon, a.b \}$ is a state cover of $M_{ex}$.

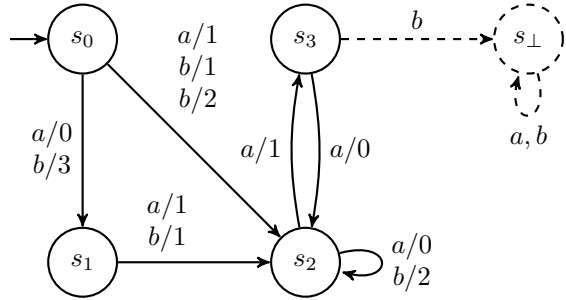
\begin{figure}[ht]
    \begin{center}
        \scalebox{1.0}{
            \begin{tikzpicture}[->,>=stealth',auto,node distance=2cm,thick]    
    \node[state,initial by arrow, initial text={}]         (q0) [] {$s_0$};
    \node[state]         (q1) [below=2cm of q0 ] {$s_1$};
    \node[state]         (q2) [right=2cm of q1 ] {$s_2$};
    \node[state]         (q3) [right=2cm of q0 ] {$s_3$};
    \node[state,dashed]         (qB) [right=2cm of q3 ] {$s_\bot$};
    
    \path 
    (q0)  edge    [align=center]  node[xshift=-10pt,yshift=10pt] {$a/1$ \\ $b/1$ \\ $b/2$} (q2)
    (q0)  edge    [align=center, swap]  node {$a/0$ \\ $b/3$} (q1)
    (q1)  edge    [align=center]  node[xshift=-10pt] {$a/1$ \\ $b/1$} (q2)
    (q2)  edge    [loop right,align=center]  node {$a/0$ \\ $b/2$} (q2)
    (q2)  edge    [bend left=10]  node {$a/1$} (q3)
    (q3)  edge    [bend left=10]  node {$a/0$} (q2)
    (q3)  edge    [dashed]  node {$b$} (qB)
    (qB)  edge    [loop below, dashed]  node {$a,b$} (qB)
    ;
\end{tikzpicture}
        }        
    \end{center}
    \caption{Example FSM $M_{ex}$. The additional state $s_\bot$ and transitions for undefined inputs used in the reachability analysis are rendered using dashed lines.}
    \label{fig:syntheticExample}
\end{figure}

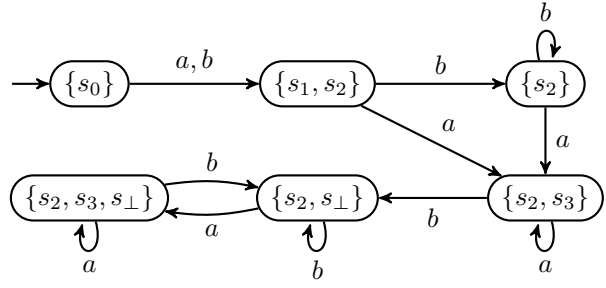
\begin{figure}[ht]
    \begin{center}
        \scalebox{1.0}{
            \begin{tikzpicture}[->,>=stealth',auto,node distance=3cm,thick]    
    \node[draw, rounded rectangle,initial by arrow, initial text={}]         (0) [] {$\{s_0\}$};
    \node[draw, rounded rectangle]         (12) [right=of 0.center, anchor=center ] {$\{s_1,s_2\}$};
    \node[draw, rounded rectangle]         (2) [right=of 12.center, anchor=center ] {$\{s_2\}$};
    \node[draw, rounded rectangle]         (23) [below=1.5cm of 2.center, anchor=center ] {$\{s_2,s_3\}$};
    \node[draw, rounded rectangle]         (2B) [left=of 23.center, anchor=center ] {$\{s_2,s_\bot\}$};
    \node[draw, rounded rectangle]         (23B) [left=of 2B.center, anchor=center ] {$\{s_2,s_3,s_\bot\}$};

    \path 
    (0)  edge    []  node {$a,b$} (12)
    (12)  edge    []  node {$b$} (2)
    (12)  edge    []  node {$a$} (23)
    (2)  edge    [loop above]  node {$b$} (2)
    (2)  edge    []  node {$a$} (23)
    (23)  edge    [loop below]  node {$a$} (23)
    (23)  edge    []  node {$b$} (2B)
    (2B)  edge    [loop below]  node {$b$} (2B)
    (2B)  edge    [bend left=10]  node {$a$} (23B)
    (23B)  edge    [loop below]  node {$a$} (23B)
    (23B)  edge    [bend left=10]  node {$b$} (2B)
    ;
\end{tikzpicture}
        }        
    \end{center}
    \caption{Determinised reachability automaton for $M_{ex}$}
    \label{fig:syntheticExampleDet}
\end{figure}

In the following we will write $\widehat{S'}$ to denote the set of all d-reachable states in some  state set 
$S' \subseteq S$. For a given state cover $V$ of $M$, we assign to each $s \in \widehat{S}$ a unique sequence $\bar{v}_s \in V$ that d-reaches $s$. Furthermore, we write $V'$ to denote the set of all responses of $M$ to $V$, that is $V' = \{ \bar{x}/\bar{y} \in \lang(M)~|~\bar{x} \in V \}$.

For the CR state machine from Fig.~\ref{fig:crbehavior}, such an assignment of unique d-reaching sequences can be chosen as given in Table~\ref{table:card-reader-state-cover}, resulting in a state cover $V_\crd = \{\bar{v}_s~|~s \in S_\crd \}$. The state cover of CR contains {\it every} state, because each state is d-reachable 
(therefore, $\widehat{S_\crd} = S_\crd$), as can
be easily seen from inspection of the FSM diagram in Fig.~\ref{fig:crbehavior}.

\begin{table}[htp]
    \footnotesize
    \caption{Elements of a state cover of the CR state machine}
    \begin{center}
        \begin{tabular}{|l|l|}\hline\hline
            {\bf State $s$} & {\bf d-reaching sequence $\bar{v}_s$} \\\hline\hline
            init & $\epsilon$ \\
            card0 & (pr.a) \\
            card1 & (pr.A) \\
            auth0 & (pr.a).(ci.in.v)\\
            auth1 & (pr.A).(ci.in.v)\\
            PIN0 & (pr.A).(ci.in.v).(ts.in.ok)\\
            PIN1 & (pr.A).(ci.in.v).(ts.in.ok).(ts.in.ip)\\
            PIN2 & (pr.A).(ci.in.v).(ts.in.ok).(ts.in.ip).(ts.in.ip)\\
            ejected0  & (ci.in.i)\\
            ejected1  & (pr.A).(ci.in.v).(ts.in.ok).(ts.in.vp)\\
            \hline\hline
        \end{tabular}
    \end{center}
    \label{table:card-reader-state-cover}
    \normalsize
\end{table}%

\subsection{Reliably Distinguishable States}\label{sec:rdis}


We will now introduce the first significant change of the classical state counting method: this concerns the distinguishability of states. 
In classical state counting, it is necessary to apply at least a single input $x$ to distinguish states $s_1, s_2$, where $x$ \textit{reliably} distinguishes the states if the sets of outputs observed on applying $x$ to $s_1$ and $s_2$ are disjoint. 
This distinction is reliable in the sense that it does not depend on the occurrence of a nondeterministic output. The concept of reliable distinguishability (\textit{r-distinguishability}) is then extended inductively such that $s_1,s_2$ are reliably distinguishable by input sequences up to length $(k+1)$ if they can be reliable distinguished by a single input or if there exists some input $x$ such that for all responses $y$ observed for both $s_1$ and $s_2$ to $x$, the states reached by applying $x/y$ to $s_1$ and $s_2$ are reliably distinguishable by input sequences up to length $k$.
Our definition takes into account that two states are immediately distinguishable if their accepted inputs $\Delta_M(s_1), \Delta_M(s_2)$ differ. This distinction is again reliable, as no nondeterminism is encountered. It can thus be integrated into the classical definition of reliable distinguishability, possibly reducing the number of inputs required to distinguish states.
These considerations lead to the following definition.

\begin{definition}\label{def:rdis}
    States $s_1$ and $s_2$ of $M$ are \textit{r(0)-distinguishable} if $\definedInputs_M(s_1) \not= \definedInputs_M(s_2)$ holds. States $s_1$ and $s_2$ of $M$ are \textit{r(k+1)-distinguishable} for $k \geq 0$, if they are r(k)-distinguishable or if there exists some $x \in \definedInputs_M(s_1) \cap \definedInputs_M(s_2)$ such that $out(s_1,x) \cap out(s_2,x) = \varnothing$ or for all $y \in out(s_1,x) \cap out(s_2,x)$  it holds that $s_1\after x/y$ and $s_2\after x/y$ are r(k)-distinguishable. States $s_1$ and $s_2$ of $M$ are \textit{r-distinguishable} if there exists some $k \in \mathbb{N}$ such that $s_1$ and $s_2$ are r(k)-distinguishable.
\end{definition}

Note that this definition of r-distinguishability extends the original definition given in \cite{hierons_testing_2004} by a new base case of r(0)-distinguishability for states that differ in their defined inputs, as motivated above. The base case of the definition used in \cite{hierons_testing_2004}, r(1)-distinguishability, is retained, as the extended definition considers states r(1)-distinguishable if they are r(0)-distinguishable or if there exists some input defined in both states for which the states generate disjoint sets of outputs.

This definition immediately leads to the following notion of sets of input sequences whose application is sufficient to establish r-distinguishability:

\begin{definition}\label{def:rdisset}
    Let $W \subseteq \Sigma_I^*$ be some set of input sequences. Then $W$ \textit{r(0)-distinguishes} any pair of r(0)-distinguishable states of $M$. Furthermore, for $k \geq 0$, $W$ \textit{r(k+1)-distinguishes} states $s_1$ and $s_2$ of $M$ if $W$ already r(k)-distinguishes them or if there exist some $x \in \definedInputs_M(s_1) \cap \definedInputs_M(s_2) \cap \Pref(W)$ such that for all $y \in out(s_1,x) \cap out(s_2,x)$ there exists some $W'$ such that $\{x\}.W' \subseteq \Pref(W)$ holds and $W'$ r(k)-distinguishes $s_1\after x/y$ and $s_2\after x/y$.
\end{definition}

It is trivial to see that the following properties are equivalent:
\begin{enumerate}
\item States $s_1$ and $s_2$ are r-distinguishable according to Definition~\ref{def:rdis}.
\item There exists $W\subseteq \Sigma_I^*$ r-distinguishing $s_1$ and $s_2$ according to Definition~\ref{def:rdisset}.
\end{enumerate}

FSM $M_{ex}$ given in Fig.~\ref{fig:syntheticExample} exhibits several distinct examples of r-distinguishability. For example, state $s_3$ can be r(0)-distinguished from any other state, as $s_3$ is the only state in which input $b$ is not defined. Furthermore, states $s_1$ and $s_2$ are r(1)-distinguishable, as there exists no output produced by both states in response to $b$. Next, states $s_0$ and $s_2$ are r(2)-distinguishable using input $a$, as $out(s_0,x) \cap out(s_2,a) = \{0,1\}$ and the states reached from $s_0$ and $s_2$ via $a/0$ and $a/1$ are respectively r(1)-distinguishable as described above. Finally, states $s_0$ and $s_1$ are not r-distinguishable, as for any input $x \in \{a,b\}$ both states reach $s_2$ via $x/1$ and no state can be r-distinguished from itself.

R-distinguishing sets can be computed based on the inductive definition of r-distinguishing sets, as described for complete FSMs in~\cite{petrenko_testing_1996} and, using adaptive tests, \cite{petrenko_testing_2011}. Function $\textsc{CollectRDSets}(M)$, described in Fig.~\ref{alg:CollectRDSets}, provides an extension of such algorithms that also considers  r(0)-distinguishability in computing pairs of state pairs and sets of input sequences such that $(\{s_1,s_2\},W)$ is contained in the return value only if $s_1$ and $s_2$ are r-distinguished by $W$ in $M$.

To do so, the algorithm  first initialises set $R$ in line 2 by  assigning to each pair of r(0)-distinguishable states the r-distinguishing empty set $W = \varnothing$.
The  $\mathbb{P}(S)$-valued auxiliary variable $P$ contains all state pairs to which no r-distinguishing set has been assigned already. Consequently, $P$ is initialised in line~3 to contain all state pairs that have not been
captured in $R$, since they are not $r(0)$-distinguishable.

Thereafter, the algorithm iteratively checks for any pair $\{s_1,s_2\} \in P$,   whether there exists some input $x$ defined in both states such that for all $y \in out(s_1,x) \cap out(s_2,x)$ there exists some $W_y$ assigned to $\{s_1\text{-after-}x/y,s_2\text{-after-}x/y\}$ in $R$. If such an $x$ exists, then an r-distinguishing set for $s_1$ and $s_2$ is created in line 12 by extending $x$ with $W_y$ for all $y \in out(s_1,x) \cap out(s_2,x)$ and assigning the resulting set to $\{s_1,s_2\}$ in $R$.\footnote{The fact that this resulting set is indeed r-distinguishing $s_1$ and $s_2$ can be shown via induction on the number of previous iterations, using as base case the assignment of $\varnothing$ to all r(0)-distinguishable pairs.} If no such $x$ exists, then $\{s_1,s_2\}$ is to be considered again in the next iteration. 

The algorithm terminates returning $R$, if all pairs of states of $M$ have been assigned some r-distinguishing set (in this case, auxiliary variable $P$ is empty), or if in some iteration no r-distinguishing set could be assigned to any element of $P$, in  which case the remaining pairs in $P$ are not r-distinguishable. Thus, if $s_1$ and $s_2$ are r-distinguishable in $M$, then there exists some $W$ such that $(\{s_1,s_2\},W)$ is contained in the return value of $\textsc{CollectRDSets}(M)$.
Note that this algorithm always terminates, as $P$ is finite and each iteration after which the algorithm does not immediately terminate must remove at least one element of $P$.

\begin{figure}
    \begin{algorithmic}[1]
        \Function{CollectRDSets}{$M= (S,\ii s,\Sigma_I,\Sigma_O,h_M)$} : $\mathbb{P} (\mathbb{P}(S)\times \mathbb{P}(\Sigma_I^*))$
        \State $R \gets \{ (\{s_1,s_2\},\varnothing)|s_1,s_2 \in S \wedge \definedInputs_M(s_1) \neq \definedInputs_M(s_2)\}$
        \State $P \gets \{ \{s_1,s_2\}~|~s_1,s_2 \in S \wedge (\{s_1,s_2\},\varnothing) \notin R  \}$
        \State $\text{\it changed} \gets \text{\it True}$ 
        \While{$P \neq \varnothing \wedge \text{\it changed}$}
            \State $\text{\it changed} \gets \text{\it False}$ \Comment{No change in $P$ yet} 
            \ForAll{$\{s_1,s_2\} \in P$}
                \State $X \gets \{x \in \definedInputs_M(s_1) \cap \definedInputs_M(s_2)~|$ \Statex [5] $\quad \ \forall y \in out(s_1,x) \cap out(s_2,x).$
                \Statex [6] $\quad \ \exists W. (\{s_1\text{-after-}x/y,s_2\text{-after-}x/y\},W) \in R\}$
                \If{$X \neq \varnothing$}
                    \State choose any $x \in X$
                    \State $W' \gets \{W~|~\exists y \in out(s_1,x) \cap out(s_2,x) :$ 
                    \Statex [7] $\quad \ \ (\{s_1\text{-after-}x/y,s_2\text{-after-}x/y\},W) \in R\}$
                    \State $R \gets R \cup \{ \big(\{s_1,s_2\}, \{x\}.(\bigcup_{W\in W'}W)\big) \}$
                    \State $P \gets P - \{ \{s_1,s_2\} \}$
                    \State $\text{\it changed} \gets \text{\it True}$ \Comment{$P$ has changed}
                \EndIf
            \EndFor
        \EndWhile
        \State \Return $R$
        \EndFunction
    \end{algorithmic}
    \caption{An algorithm to compute r-distinguishing sets for all pairs of r-distinguishable states of an FSM}\label{alg:CollectRDSets}
\end{figure}

The following lemma justifies the use of  any $W$ that r-distinguishes states of $M$ reached by a pair of IO sequences to distinguish  the states of $I$ reached by the same sequences, if no failure is uncovered by applying  $W$:

\begin{lemma}\label{lem:rdist_dist}
Let $k \in \mathbb{N}$ and suppose that $W\subseteq \Sigma_I^*$ does r(k)-distinguish $\ii s\after \bar{x}_1/\bar{y}_1$ and $\ii s\after \bar{x}_2/\bar{y}_2$ for some $\bar{x}_1/\bar{y}_1, \bar{x}_2/\bar{y}_2 \in \lang(M) \cap \lang(I)$. Then, if $I$ passes $\{\bar{x}_1,  \bar{x}_2\}.W$, traces $\bar{x}_1/\bar{y}_1$ and $\bar{x}_2/\bar{y}_2$ reach distinct states in $I$.
\end{lemma}
\begin{proof}
    Let $s_i = \ii s\after \bar{x}_i/\bar{y}_i$ and $t_i = \ii t\after \bar{x}_i/\bar{y}_i$ for $i \in \{1,2\}$.   
    We prove the desired result by induction on $k$. First assume that $k=0$. Then, as $I$ passes $\{\bar{x}_1, \bar{x}_2\}.W$ and thus in particular $\{\bar{x}_1, \bar{x}_2\}$, $\definedInputs_I(t_1) = \definedInputs_M(s_1) \not= \definedInputs_M(s_2) = \definedInputs_I(t_2)$ and hence $t_1 \not= t_2$.
    
As induction hypothesis, assume that the lemma holds for all $k = 0,\dots,k'$ with $k'\ge 0$.
    
    For the induction step let $k = k' + 1$ and assume that $s_1$ and $s_2$ are not r(0)-distinguishable, as this case would be identical to the base case. This implies that $\definedInputs_M(s_1) = \definedInputs_M(s_2) \not= \varnothing$, as $s_1$ and $s_2$ are r-distinguishable. Then, as $W$ r(k)-distinguishes $s_1$ and $s_2$, there must exist some $x \in \definedInputs_M(s_1) \cap \definedInputs_M(s_2) \cap \Pref(W)$ such that for all $y \in out(s_1,x) \cap out(s_2,x)$ there exists some $W'$ such that $\{x\}.W' \subseteq \Pref(W)$ holds and $\{x\}.W'$ does r(k')-distinguish $s_1\after x/y$ and $s_2\after x/y$. 
    If $out(s_1,x) \cap out(s_2,x)$ is empty, then $t_1 \not= t_2$ follows from $W$ containing some sequence $x.\bar{x}'$ and $I$ passing $\{\bar{x}_1,\bar{x}_2\}.\{x.\bar{x}'\}$, which requires $out(t_i,x) \subseteq out(s_i,x)$ to hold for $i \in \{1,2\}$.
    Thus let $y$ be an arbitrary element of $out(s_1,x) \cap out(s_2,x)$ and let $s_i' = s_i\after x/y = \ii s\after \bar{x}_i.x/\bar{y}_i.y$ for $i \in \{1,2\}$. By the properties of $W$ and $x$ there must exist some $W'$ such that $\{x\}.W' \subseteq \Pref(W)$ holds and $W'$ r(k')-distinguishes $s_1'$ and $s_2'$. Then, by the induction hypothesis, $\ii t\after \bar{x}_1.x/\bar{y}_1.y$ and $\ii t\after \bar{x}_2.x/\bar{y}_2.y$ reach distinct states in $I$, which implies the desired inequality $t_1 \not= t_2$ due to $I$ being observable.
\end{proof}

In the following we will use $S_D \subseteq \mathbb{P}(S)$ to denote the set of maximal sets of pairwise r-distinguishable states of $M$. For every $s \in S$ that is not r-distinguishable from any other state in $S$, $S_D$ includes a singleton set $\{s\}$. Therefore, every state in $S$ is contained in some element of $S_D$.

Note here that calculating this set is identical to finding all maximal cliques in the undirected graph whose vertices are the state of $M$ and where two states are adjacent if and only if they are r-distinguishable. This constitutes a computationally expensive problem, as described, for example,  in~\cite{TOMITA200628}. Should this computation be unfeasible for some large FSM, then it is also sufficient to instead only use a subset of $S_D$, as long as each state of the FSM is contained in some element of this subset, but such a reduction might delay the termination of algorithms described later.

Using Table~\ref{table:acc} it is easy to see that in the CR state machine many pairs of states are r(0)-distinguishable due to differing defined inputs. That is, the states of the CR state machine can be partitioned into four sets based on their disabled inputs:
\begin{align*}
    G_1 &:= \{\texttt{init, card0, card1}\} \\
    G_2 &:= \{\texttt{auth0, auth1}\} \\
    G_3 &:= \{\texttt{ejected0, ejected1}\} \\
    G_4 &:= \{\texttt{PIN0, PIN1, PIN2}\}
\end{align*}
such that for all $1 \leq i < j \leq 4$ it holds that each pair of states $s_i \in G_i, s_j \in G_j$ is r(0)-distinguishable and hence r-distinguished by any set of input sequences. Furthermore, state \texttt{init} can be r(1)-distinguished from any other state by application of \texttt{pr.a} or \texttt{pr.A}, as it is the only state not ignoring these inputs. Neither the two \texttt{card}-states nor the two \texttt{auth}-states are r-distinguishable, as they differ in behaviour only by the additional transition from \texttt{auth0} to \texttt{ejected1} on \texttt{ts.in.ok}. This  cannot be used to r-distinguish \texttt{auth0} and \texttt{auth1}, as both of these states also reach \texttt{PIN0} on input \texttt{ts.in.ok} with the same output {\tt ts.out.p}. Next, the \texttt{PIN}-states can be r-distinguished by one or two applications of \texttt{ts.in.ip}. Finally, the \texttt{ejected}-states can be r(1)-distinguished via \texttt{ci.r}.
Thus, there exist four maximal sets of pairwise r-distinguishable states of the CR state machine:
\begin{align*}
    S_{00} &:= \{\texttt{card0, auth0}\} \cup \{\texttt{init}\} \cup G_3 \cup G_4 \\
    S_{01} &:= \{\texttt{card0, auth1}\} \cup \{\texttt{init}\} \cup G_3 \cup G_4 \\
    S_{10} &:= \{\texttt{card1, auth0}\} \cup \{\texttt{init}\} \cup G_3 \cup G_4 \\
    S_{11} &:= \{\texttt{card1, auth1}\} \cup \{\texttt{init}\} \cup G_3 \cup G_4
\end{align*}

\subsection{Test Suite Generation}\label{sec:gen}

The test suite generation algorithm described in this section is the second significant change in comparison to the classical state counting method. By taking the information $\Delta_M(s)$ about accepted events in states $s$ into account, we can save a substantial number of test cases. 

Throughout this section, we assume that the reference model $M$ is represented in such a form that 
$$
\forall y\in \Sigma_O, s_2\in S.\ (s_1,x,y,s_2)\not\in h_M
$$
implies that input $x$ is disabled in state $s_1$. Ignored events are supposed to be always present in $h$ with self-loop transitions and associated null-output. For our example from Section~\ref{sec:crexample}, this means that self-loop transitions labelled by ${\tt pr.a/null}$ and ${\tt pr.A/null}$
are added in Fig.~\ref{fig:crbehavior} to all states but {\tt init}. All other unhandled inputs in this diagram   indicate disabled inputs.

The test strategy described in \cite{hierons_testing_2004} creates a test suite by iterative extension of sequences after each $\bar{v}_s \in V$, using a termination criterion based on state counting. For each $\bar{v}_s$, this extension process  starts with $\{\epsilon\}$. In each iteration, the process extends a previously considered sequence $\bar{x}$ only if there exists some $\bar{y}$ such that $\bar{x}/\bar{y} \in \lang_M(s)$ and for all $S' \in S_D$ it holds that the nonempty prefixes of $\bar{x}/\bar{y}$ applied to $s$ reach states of $S'$ at most $m-|\widehat{S'}|$ times  (recall that $\widehat{S'}$ denotes the subset of d-reachable states from state set $S'$).
For any $s \in \widehat{S}$, $\bar{x}/\bar{y} \in \lang_M(s)$ and $S' \in S_D$, we say that $S'$ \emph{terminates} $\bar{x}/\bar{y}$ for $s$ and $m$, if  the nonempty prefixes of $\bar{x}/\bar{y}$ applied to $s$ reach    states of $S'$ exactly $m-|\widehat{S'}|+1$ times and no proper prefix of $\bar{x}/\bar{y}$ is terminated for $s$ and $m$ by any element of  $S_D$. We denote the set of all $S' \in S_D$ that terminate $\bar{x}/\bar{y}$ for $s$ and $m$ by $term(s,\bar{x}/\bar{y},m)$. The result of the iterative extension process for some $s \in \widehat{S}$ can then be defined as 
\begin{align*}
&Tr(s,m) := \Pref \{ \bar{x}~|~\exists \bar{y} : \ \bar{x}/\bar{y} \in \lang_M(s) \wedge \ term(s,\bar{x}/\bar{y},m) \not= \varnothing \}
\end{align*}

Note here that this iterative extension process always terminates, as each state $s \in S$ is contained in at least one set   $S'\in S_D$. Thus, each extension $\bar{x}.x/\bar{y}.y \in \lang_M(s)$  of some trace $\bar{x}/\bar{y}$ visits at least one element of $S_D$ an additional time compared to $\bar{x}/\bar{y}$. Therefore, as $S_D$ is finite, no trace can be extended infinitely without being terminated.

Continuing example $M_{ex}$ from Fig.~\ref{fig:syntheticExample}, Fig.~\ref{fig:syntheticExampleTr} shows the extension process of calculating $Tr(s_0,4)$ such that the maximal paths in the tree constitute the set of all $\bar{x}/\bar{y} \in \lang_M(s_0)$ that are terminated by at least one of the two maximal sets of pairwise distinguishable states of $M_{ex}$, namely $\{s_0,s_2,s_3\}$ and $\{s_1,s_2,s_3\}$. The resulting input projection of these paths then constitutes $Tr(s_0,4)$, which in this case is the set of all input sequences of length 3 or 4 over alphabet $\{a,b\}$.
\begin{figure}[ht]
    \begin{center}
        \scalebox{0.8}{
            \begin{tikzpicture}[->,>=stealth',auto,node distance=0cm,thick]    
    \node[state,initial by arrow, initial text={}]         (root) [] {$s_0$};
    \node[state]         (l) [below left=1.2cm and 1.25cm of root ] {$s_1$};

    \node[state]         (r) [below right=1.2cm and 1.25cm of root ] {$s_2$};
    \node[state]         (rl) [below left=1.2cm and 0.5cm of r ] {$s_2$};

    \node[state]         (ll) [left=1.75cm of rl ] {$s_2$};
    \node[state]         (lll) [below left =1.2cm and 0.5cm of ll ] {$s_2$};
    \node[state]         (llll) [below left =1.2cm and 0.00cm of lll ] {$s_2$};
    \node[state]         (lllr) [below right =1.2cm and 0.00cm of lll ] {$s_3$};
    \node[state]         (llr) [below right =1.2cm and 0.5cm of ll ] {$s_3$};
    \node[state]         (llrl) [right=0.74cm of lllr ] {$s_2$};

    \node[state]         (rr) [below right =1.2cm and 0.5cm of r ] {$s_3$};
    \node[state]         (rll) [right=0.64cm of llr ] {$s_2$};
    \node[state]         (rrl) [right =1.35cm of rll ] {$s_2$};

    \path 
    (root)  edge    [align=center, swap]  node {$a/0$ \\ $b/3$} (l)
    (root)  edge    [align=center]  node {$a/1$ \\ $b/1$ \\ $b/2$} (r)
    (l)  edge    [align=center, swap]  node {$a/1$ \\ $b/1$} (ll)
    (ll)  edge    [align=center, swap]  node[xshift=-0.1cm, yshift=-0.2cm] {$a/0$ \\ $b/2$} (lll)
    (ll)  edge    [align=center]  node {$a/1$} (llr)
    (lll)  edge    [align=center, swap]  node[xshift=-0.1cm, yshift=-0.2cm] {$a/0$ \\ $b/2$} (llll)
    (lll)  edge    [align=center]  node {$a/1$} (lllr)
    (llr)  edge    [align=center, swap]  node {$a/0$} (llrl)
    (r)  edge    [align=center, swap]  node[xshift=-0.1cm, yshift=-0.2cm] {$a/0$ \\ $b/2$} (rl)
    (r)  edge    [align=center]  node {$a/1$} (rr)
    (rl)  edge    [align=center, swap]  node {$a/0$ \\ $b/2$} (rll)
    (rr)  edge    [align=center, swap]  node {$a/0$} (rrl)
    ;

    \draw [
        -,
        thick,
        decoration={
            brace,
            mirror,
            raise=0.6cm
        },
        decorate
    ] (llll.west) -- (llrl.east) 
    node [pos=0.5,anchor=north,yshift=-0.7cm] {$\{\{s_0,s_2,s_3\}, \{s_1,s_2,s_3\}\}$}; 

    \draw [
        -,
        thick,
        decoration={
            brace,
            mirror,
            raise=0.6cm
        },
        decorate
    ] (rll.west) -- (rrl.east) 
    node [pos=0.5,anchor=north,yshift=-0.7cm] {$\{\{s_0,s_2,s_3\}\}$}; 
\end{tikzpicture}
        }        
    \end{center}
    \caption{Graphical representation of the extension process for $Tr(s_0,4)$ in $M_{ex}$ as a tree which merges sequences visiting the same sequence of states in $M_{ex}$. Nodes indicate reached states, while the sets given in brackets below the leaves indicate $term(s_0,\bar{x}/\bar{y})$ for all $\bar{x}/\bar{y}$ reaching these leaves.}
    \label{fig:syntheticExampleTr}
\end{figure}
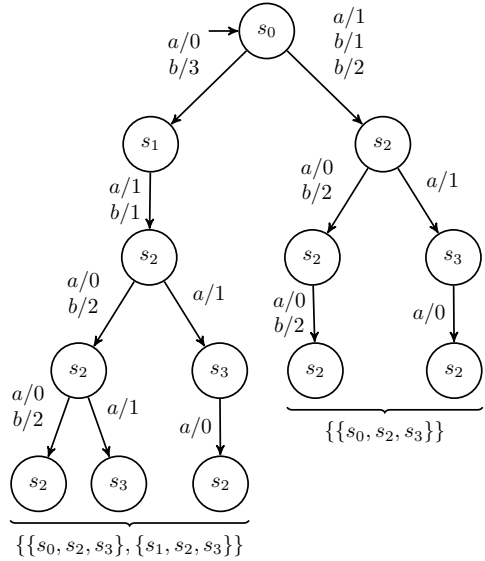

Finally, the test suite is constructed by applying for each $s \in \widehat{S}$ all sequences in $Tr(s)$ after $\bar{v}_s$ and by applying r-distinguishing sets based on the termination criterion. That is, for any $\bar{x}/\bar{y}$ terminated for $s$ and $m$, some $S_i$ is selected that terminates it and then after each pair of distinct sequences $\bar{x}_1/\bar{y}_1, \bar{x}_2/\bar{y}_2 \in V' \cup (\{\bar{v}_s\}.\Pref(\bar{x}/\bar{y}))$ that reach distinct states $s_1, s_2$ in $S_i$, some set $W$ is selected that r-distinguishes those states, and the sequences $\{\bar{x}_1, \bar{x}_2\}.W$ are added to the test suite. Finally, all proper prefixes of the test suite are removed. 

This test strategy is implemented in function $\textsc{GenerateTestSuite}(M,m)$ detailed in Fig.~\ref{alg:GenerateTestSuite}, which computes a test suite $T \subset \Sigma_I^*$ for specification $M$ and upper bound $m$ of the size of FSMs in the fault domain $\faultDomain$. A result of applying the strategy to $M_{ex}$ for $m=4$ is given in Appendix~\ref{appendix:mex_test_suite}.

The following two lemmata establish the soundness and exhaustiveness of any test suite generated by this strategy.

\begin{figure}
    \begin{algorithmic}[1]
        \Function{GenerateTestSuite}{$M,m$}  : $\mathbb{P}(\Sigma_I^*)$
        \State{choose a state cover $V$ of $M$}
        \State{$V' \gets \{ \bar{x}/\bar{y} \in \lang(M)~|~\bar{x} \in V \}$}
        \State{calculate $Tr(s,m)$ for each $s\in \widehat{S}$}
        \State{$T\gets \bigcup_{s \in \widehat{S}} \  \{\bar{v}_s\}.Tr(s,m)$}
        \State $D\gets \{ (s,\bar{x}/\bar{y})~|~s \in \widehat{S}, \bar{x}/\bar{y} \in \lang_M(s), \ term(s,\bar{x}/\bar{y},m) \not= \varnothing  \}$ 
        \ForAll{$(s,\bar{x}/\bar{y}) \in D$}
            \State{choose an $S_i \in term(s,\bar{x}/\bar{y},m)$}
            \ForAll{$\bar{x}_1/\bar{y}_1, \bar{x}_2/\bar{y}_2 \in V' \cup (\{\bar{v}_s\}.\Pref(\bar{x}/\bar{y}))$}
                \State{$s_1 \gets \ii s\after \bar{x}_1/\bar{y}_1$}
                \State{$s_2 \gets \ii s\after \bar{x}_2/\bar{y}_2$}
                \If{$s_1 \not= s_2 \wedge s_1 \in S_i \wedge s_2 \in S_i$}
                    \State{$W' \gets \{ \bar{x}'~|~\{\bar{x}_1.\bar{x}',\bar{x}_2.\bar{x}'\} \subseteq \Pref(T) \}$}
                    \If{$W'$ does not r-distinguish $s_1, s_2$}
                        \State{choose $W$ that r-distinguishes $s_1, s_2$}
                        \State{$T \gets T \cup \{\bar{x}_1, \bar{x}_2\}.W$}
                    \EndIf
                \EndIf
            \EndFor
        \EndFor
        \State{$T \gets \{ \bar{x} \in T~|~\nexists \bar{x}' : \bar{x}' \not= \epsilon \wedge \bar{x}.\bar{x}' \in T  \} $}
        \State \Return $T$
        \EndFunction
    \end{algorithmic}
    \caption{Algorithm generating $m$-complete $\strongred$-conformance test suites.}\label{alg:GenerateTestSuite}
\end{figure}

\begin{lemma}\label{lem:sound}
    Let $T$ be a test suite generated by $\textsc{GenerateTestSuite}(M,m)$. Then $T$ is \textit{sound}: For any $I \in \faultDomain$, if $I \strongred M$ holds, then $I$ passes $T$.
\end{lemma}
\begin{proof}
Let $\bar{x}$ be an input sequence in $T$ and assume that $I \strongred M$ holds but $I$ fails $\bar{x}$ and thus $T$.
Suppose that $\bar{x}$ is the empty input sequence. Then the test case can only fail because $\Delta_I(\ii t) \neq \Delta_M(\ii s)$, and this contradicts the assumption that $I \strongred M$.

If $\bar{x}$ has positive length, then there must exist some prefix $\bar{x}'.a'$ of $\bar{x}$ and some $\bar{y}'.z'\in\Sigma_O^*$, such that
$\bar{x}'.a'/\bar{y}'.z'\in L(I)$ and
$\bar{x}'/\bar{y}'\in L(I)\cap L(M)$ and $\Delta_I(\ii t\after \bar{x}'/\bar{y}') = \Delta_M(\ii s\after \bar{x}'/\bar{y}')$ (so the test has not yet failed), but either (1) $\bar{x}'.a'/\bar{y}'.z'\not\in L(M)$ or (2) $\Delta_I(\ii t\after\bar{x}'.a'/\bar{y}'.z') \neq \Delta_M(\ii s\after\bar{x}'.a'/\bar{y}'.z')$. Case (1) contradicts the fact that $I$ is a reduction of $M$, and Case (2) contradicts the fact that, as a {\it strong} reduction, $I$ always needs to accept exactly the same inputs as $M$. This completes the proof.
\end{proof}

\begin{lemma}\label{lem:exhaustive}
    Let $T$ be a test suite generated by $\textsc{GenerateTestSuite}(M,m)$. Then $T$ is \textit{exhaustive}: For any $I \in \faultDomain$, if $I$ passes $T$, then $I \strongred M$ holds.
\end{lemma}
\begin{proof}
    Assume that $I$ passes $T$ though $I \strongred M$ does {\it not} hold. Consider first the special case
that     $I \not\strongred M$ because $\Delta_I(\ii t) \neq \Delta_M(\ii s)$. This means that the first violation of the strong reduction relation already occurs in the initial state $\ii t$ of the implementation, without providing any input. Since the grey-box testing assumption provides $\Delta_I(\ii t)$ and the test oracle operates by comparing the SUT behaviour to the model $M$ in back-to-back fashion (see Section~\ref{sec:oracles}), this error will be immediately revealed, regardless of the test suite applied. Therefore, we assume for the remainder of this proof that $\Delta_I(\ii t) = \Delta_M(\ii s)$, so that detection of the conformance violation requires an input sequence of minimal length~1.

    Let $V, V'$ be defined as in lines 2 and 3 of the algorithm.
    Then, as $V$ contains $\epsilon$, there must exist a minimal length sequence $\bar{x}/\bar{y}$ such that some $s \in \widehat{S}$ and $\bar{v}_s/\bar{v}_s' \in V'$ exist such that $\definedInputs_I(\bar{v}_s.\bar{x}/\bar{v}_s'.\bar{y}) \not= \definedInputs_M(\bar{v}_s.\bar{x}/\bar{v}_s'.\bar{y})$ or $\bar{v}_s.\bar{x}/\bar{v}_s'.\bar{y} \in \lang(I) \setminus \lang(M)$ holds. 
    Therefore, input sequence $\bar{v}_s.\bar{x}$ cannot be contained in $T$, as $I$ passes $T$, and hence there must exist a proper prefix $\bar{x}'/\bar{y}'$ of $\bar{x}/\bar{y}$ such that $(s,\bar{x}'/\bar{y}') \in D$, where $D$ is the set of all $(s,\bar{x}/\bar{y})$ such that $\bar{x}/\bar{y}$ is terminated for $s \in \widehat S$ and $m$, as assigned in line~6 of the algorithm. The for-loop in line~7 will then run through one cycle where
the pair     $(s,\bar{x}'/\bar{y}')$ is processed.
    
    In this cycle, let $S_i \in term(s,\bar{x}'/\bar{y}',m)$ denote the set chosen in line~8 of the algorithm. 
Let 
$
P = \{ \bar{x}''/\bar{y}'' \in \Pref(\bar{x}'/\bar{y}') \setminus \{\epsilon\}~|~s\after \bar{x}''/\bar{y}'' \in S_i\}
$
denote the set of all nonempty prefixes of $\bar{x}'/\bar{y}'$ that reach states of $S_i$ if applied to $s$. By construction, $\{\bar{v}_s/\bar{v}_s'\}.P$ must reach states of $S_i$ exactly $m-|\widehat{S_i}|+1$ times and hence $|P| = m-|\widehat{S_i}|+1$. Let $P = \{\tau_1,\ldots,\tau_{m-|\widehat{S_i}|+1}\}$, where $\tau_i$ is a proper prefix of $\tau_j$ for all $1 \leq i < j \leq m-|\widehat{S_i}|+1$.
    
    Next, note that $I$ must exhibit some behaviour $\bar{v}/\bar{v}' \in V' \cap \lang(I)$ for any $\bar{v} \in V$ in order to pass $V \subseteq \Pref(T)$. Let $\widehat{S_i} = \{s_1,\ldots,s_k\}$, and for each $s_i \in \widehat{S_i}$ let  $\pi_i \in V' \cap \lang(I)$ denote an arbitrary IO trace $\bar{v}_{s_i}/\bar{v}_{s_i}' \in V' \cap \lang(I)$, and finally let $V^I = \{\pi_i~|~i \in \{1,\ldots,k\}\}$. Set $V^I$ then contains $|\widehat{S_i}|$ sequences 
reaching states     $\{s_1,\ldots,s_k\} = \widehat{S_i}\subseteq S_i$ in $M$.
        
    Thus, the sequences in the disjoint union of $V^I$ and $\{\bar{v}_s/\bar{v}_s'\}.P$ reach states in $S_i$ exactly $m+1$ times. 
Therefore,   as $|I| \leq m$, there must exist some state of $I$ that is reached by two  IO-traces 
$\alpha, \beta\in \big(V^I \cup \{\bar{v}_s/\bar{v}_s'\}.P\big)$ that are either distinct or contained in both operands of the union. Regarding the containment
of $\alpha$ and $\beta$ in $V^I$ or $\{\bar{v}_s/\bar{v}_s'\}.P$, there exist three possible cases:
    \begin{align*}
    &\text{Both $\alpha$ and $\beta$ in $V^I$, that is,}  & (a)\\
    &    \qquad \exists i < j \in \{1,\ldots,|\widehat{S_i}|\}: t_0\after \pi_i = t_0\after \pi_j \\
    &     \text{Different sets: $\alpha \in V^I$ and $\beta \in \{\bar{v}_s/\bar{v}_s'\}.P$, that is,}  & (b)\\
    &    \qquad \exists i \in \{1,\ldots,|\widehat{S_i}|\}, j \in \{1,\ldots,m-|\widehat{S_i}|+1\}: \ t_0\after \pi_i = t_0\after (\bar{v}_s/\bar{v}_s').\tau_j \\
    &    \text{Both $\alpha$ and $\beta$ in $\{\bar{v}_s/\bar{v}_s'\}.P$, that is,}  & (c)\\
    &    \qquad \exists i < j \in \{1,\ldots,m-|\widehat{S_i}|+1\}: \   t_0\after (\bar{v}_s/\bar{v}_s').\tau_i = t_0\after (\bar{v}_s/\bar{v}_s').\tau_j 
    \end{align*}
    
    Consider next the properties of $\alpha$ and $\beta$ in $M$. Let $s_\alpha = \ii s\after \alpha$ and $s_\beta = \ii s\after \beta$. Suppose that $s_\alpha \neq s_\beta$, then these states are r-distinguishable, as they are both contained in $S_i$. Thus, after having executed lines 10 to 16 of the algorithm, $T$ contains an r-distinguishing set $W$ for $s_\alpha$ and $s_\beta$ that is applied after $\{\alpha,\beta\}$ and thus, by Lemma~\ref{lem:rdist_dist}, $\alpha$ and $\beta$ must reach distinct states in any $I$ passing $T$, contradicting the assumption that they reach the same state $t$ in $I$. This shows that $\alpha$ and $\beta$ reach the same state in $M$. 
    
    Next we consider cases (a) to (c):
    
    In case (a), $\alpha = \pi_i$ and $\beta = \pi_j$ reaching the same state in $M$ requires $V$ to contain two distinct input sequences reaching the same state in $M$, which contradicts the minimality requirement on state covers.
    
     In case (b), let $\alpha = \pi_i$ and $\beta = (\bar{v}_s/\bar{v}_s').\tau_j$, and note that there exists some $\tau'$ such that $\bar{x}/\bar{y} = \tau_j.\tau'$ and also $|\tau_j| \geq 1$ and hence $|\tau'| < |\bar{x}/\bar{y}|$. 
Then, as $\alpha$ and $\beta$ reach the same states both in $M$ and in $I$, and since a failure can be observed along $\tau'$ applied after $\beta$ and hence also after $\alpha$, for $\tau'$ it holds that $s_i \in \widehat{S_i}$, $\alpha \in V'$, and also
one of the following holds: $\definedInputs_I(\alpha.\tau') \not= \definedInputs_M(\alpha.\tau')$ or  $\alpha.\tau' \in \lang(I) \setminus \lang(M)$, which contradicts the minimality of $\bar{x}/\bar{y}$.
    
    In case (c), similarly to case (b), let $\alpha = (\bar{v}_s/\bar{v}_s').\tau_i$ and $\beta = (\bar{v}_s/\bar{v}_s').\tau_j$ and note that there exists  nonempty $\tau',\tau''$ such that $\tau_j = \tau_i.\tau'$ and $\bar{x}/\bar{y} = \tau_j.\tau'' = \tau_i.\tau'.\tau''$ and thus $|\tau_i.\tau''| < |\bar{x}/\bar{y}|$.
    Then, as $(\bar{v}_s/\bar{v}_s').\tau_i$ and $(\bar{v}_s/\bar{v}_s').\tau_i.\tau'$ reach the same states both in $M$ and in $I$ and as a failure can be observed along $\tau''$ applied after $(\bar{v}_s/\bar{v}_s').\tau_i.\tau'$ and hence also after $(\bar{v}_s/\bar{v}_s').\tau_i$, for $\tau_i.\tau''$ it holds that $s_i \in \widehat{S_i}$, $\bar{v}_s/\bar{v}_s' \in V'$, $(\bar{v}_s/\bar{v}_s').\tau_i.\tau' \in \lang(I) \cap \lang(M)$ and also
    one of the following holds: $\definedInputs_I((\bar{v}_s/\bar{v}_s').\tau_i.\tau'') \not= \definedInputs_M((\bar{v}_s/\bar{v}_s').\tau_i.\tau'')$ or $(\bar{v}_s/\bar{v}_s').\tau_i.\tau'' \in \lang(I) \setminus \lang(M)$, which again contradicts the minimality of $\bar{x}/\bar{y}$.
    
    Thus, in every case a contradiction can be derived. Therefore, the initial assumption 
that $I$ passes $T$ without fulfilling $I \strongred M$  cannot hold, establishing the exhaustiveness of $T$. 
\end{proof}

The $m$-completeness of test suites generated by the strategy described above then follows from Lemma~\ref{lem:sound} and Lemma~\ref{lem:exhaustive}:
\begin{theorem}
    Let $T$ be a test suite generated by $\textsc{GenerateTestSuite}(M,m)$. Then $T$ is $m$-complete: 
 For any $I \in \faultDomain$,    $I \strongred M$ holds if and only if $I$ passes $T$.
\end{theorem}

\subsection{Complexity Considerations}\label{sec:complexitymain}

In this section, a few boundary cases of the test generation algorithm specified in Fig.~\ref{alg:GenerateTestSuite} are analysed with respect to the number of test cases to be generated by the algorithm. Throughout the   section, let $a = m-n \ge 0$ denote the maximal number of {\it additional} states that may be contained in the implementation.

\subsubsection{Deterministic, completely specified case.}

For this type of FSMs, testing for strong reduction is equivalent to testing for language equivalence.
Asymptotically, considering large state spaces $n$ or large differences $a = m-n$, the bound calculated
in Appendix~\ref{sec:appdeter}, Formula~\eqref{eq:deterbound}    is simplified to
$$
      B = O(n^2\cdot|\Sigma_I|^{a+1})
$$  
This is the worst case bound for the W-Method already known from~\cite{chow:wmethod,vasilevskii1973}. It should be noted that on average, the algorithm presented here produces significantly fewer test cases than the W-Method for the deterministic, completely specified case, because it executes the H-Method algorithm which is known to produce significantly fewer   cases for complete test suites than the W-Method~\cite{DBLP:conf/forte/DorofeevaEY05}.

\subsubsection{Best-Case Test Suite Size Reduction Effect from Grey-Box Testing.}
As shown in Appendix~\ref{sec:bestcase}, the number of test cases needed for a $\strongred$-complete test suite is
bounded by
$$
n\cdot|\Sigma_I|^{a+1}
$$
if the nondeterministic reference model has d-reachable states only, and each pair of states is r(0)-distinguishable. In this case, states can simply be distinguished by their grey-box information about enabled inputs. Then   the test suite size only depends linearly on the size $n$ of the reference model.
This leads to a significant reduction of test suite size in situations, where many large r-distinguishing sets would be needed in absence of r(0)-distinguishability. For example, if for each state $s_1$ of the reference model there exists for each other state $s_2$ some input sequence r-distinguishing $s_1,s_2$, then the total number of traces in r-distinguishing sets applied to sequences reaching $s$ is already proportional to $n$\footnote{It is shown in~\cite[Section~4.6]{PeleskaHuangLectureNotesMBT}, for example, that the maximal cardinality of minimal characterisation sets in language equivalence testing is $n-1$. This also constitutes an upper bound for the minimal combined size of r-distinguishing sets applied after some sequence, if each pair of states in the reference model is r-distinguishable by a single input sequence.}.
In this case, the test suite size in a black-box setting (no r(0)-distinguishability) would be proportional to~$n^2$. This theoretical complexity result is practically confirmed by the experiments described in Section~\ref{sec:r0impact}.


It should be noted, however, that in the best case for black-box testing, a single input sequence might suffice to r-distinguish all pairs of r-distinguishable states. Then, the black-box test suite size would also grow 
linearly with $n$.

Summarising, we can say that the grey-box approach increases the number of reference models where test suite size depends only linearly on $n$.

The assumed maximal difference $a = m-n$  influences the test suite size again exponentially. 
This exponential dependency is inevitable, as explained in~\cite{DBLP:journals/pcs/El-FakihDYB12}.


\subsubsection{Worst Case}
We now consider the worst case situation leading to the maximal number of test cases necessary to achieve completeness. This occurs when the reference model is nondeterministic, the initial state is the only d-reachable state of the reference model, and no two states are reliably distinguishable.  This means that the reference model $M$ is effectively completely specified:   all states exhibit the same set of defined inputs (otherwise we would have r(0)-distinguishable states). Moreover, it is shown in Appendix~\ref{sec:worstcase}, that up to
$$
|\Sigma_I|^{mn}
$$
test cases are needed to prove  $\strongred$-completeness, which is again ``ordinary'' reduction-completeness because $M$ is completely specified. This result is consistent with the observation that testing all traces of length $mn$ can detect {\it any} 
conformance violation, when applying suitable test oracles. 
The proof of this statement is based on the investigation of 
the product FSM built from reference model $M$ and implementation model $I$; see, for example, \cite[Section~4.5]{PeleskaHuangLectureNotesMBT}.

\subsection{Applicability to Other Conformance Relations}\label{sec:relationrelations}

As described in Section~\ref{sec:ssr}, strong reduction differs from language equivalence, reduction and quasi-reduction in the implementations it considers to be conforming to a given specification. Furthermore, r-distinguishability for strong reduction, as given in Definition~\ref{def:rdis}, is weaker than r-distinguishability for (quasi-)reduction, as defined in~\cite{hierons_testing_2004,DBLP:conf/fates/PetrenkoY05}. The latter does not consider states $s_1,s_2$ to be r-distinguishable by inputs that are defined in only either $s_1$ or $s_2$. Since quasi-reduction allows conforming implementation to behave arbitrarily in response to undefined inputs, r(0)-distinguishability cannot be exploited in testing for this conformance relation. 

Test suites and associated test oracles that are complete for language equivalence, reduction, or quasi-reduction are not necessarily complete for testing strong reduction.  For example, oracles for language equivalence   require the SUT and the specification to exhibit   exactly the same sets of responses to given input sequences. 
Obviously, this would lead to unsound test suites  for strong reduction. 

If the oracle described in Section~\ref{sec:oracles} is used instead, some test suites for language equivalence, reduction or quasi-reduction may also be complete for strong reduction. The class of strategies generating such test suites includes the well-known ``brute force'' strategy based on product FSMs\footnote{This strategy has been described, for example, 
in the lecture notes~\cite[Section~4.5]{PeleskaHuangLectureNotesMBT}.}, which enumerates all input sequences of length $mn$ and thus generates prohibitively large test suites for non-trivial specifications. By the same adaptation of oracles, complete test strategies for quasi-reduction such as~\cite{DBLP:conf/fates/PetrenkoY05} can be employed to test strong reduction, since states that are r-distinguishable for quasi-reduction are also r-distinguishable for strong reduction. However, as will be discussed in Section~\ref{sec:r0impact}, this stronger definition of r-distinguishability can result in fewer pairs of states being r-distinguishable, leading to larger test suites.
These inefficiencies in adapting existing test strategies for strong reduction justify the introduction of function \textsc{GenerateTestSuite} presented in Section~\ref{sec:gen} as a new and complete test suite generation strategy for this conformance relation.

Consider next a test suite $T$ generated by \textsc{GenerateTestSuite} for some specification $M$, with use of the test oracles described in Section~\ref{sec:oracles}. This $T$ is complete for strong reduction, but in general not  complete for language equivalence, reduction, or quasi-reduction. First, $T$ is not exhaustive for language equivalence if $M$ is nondeterministic, as $T$ can be passed by implementations $I \in \faultDomain$ that only exhibit a proper subset of the language of $M$. Next, $T$ is not sound for reduction if the initial state of $M$ has defined inputs, as $\faultDomain$ contains FSMs $I$ without any transitions. Such $I$ are reductions of $M$ but do not pass $T$, as their initial states have fewer defined inputs than that of $M$. Finally, $T$ is not sound for quasi-reduction if $M$ contains states with disabled inputs, as $\faultDomain$ contains FSMs $I$ that contain all states and transitions of $M$ and also some transitions for state $s$ and input $x$ such that $x$ is disabled in $s$ in $M$. Such $I$ are quasi-reductions but do not pass $T$, as they contain states exhibiting more defined inputs than the corresponding states in reference model $M$. 

By using modified oracles, test suites generated by \textsc{GenerateTestSuite} can be used to test for language-equivalence. Such oracles must require  the SUT to exhibit the same set of responses as the specification to any test case and also check after each step that the defined inputs in the current states of SUT and specification coincide. This approach uses the fact that any set $W$ that r-distinguishes states $s_1,s_2$ must contain some input sequence $\bar{x}$ such that the responses of $s_1$ and $s_2$ to $\bar{x}$ differ or some shared response leads to state with differing defined inputs. Furthermore, \textsc{GenerateTestSuite} can be used to test for reduction between complete observable FSMs, as in this case the algorithm is essentially reduced to the classical state counting method. This also enables the use of \textsc{GenerateTestSuite} in generating test suites that are complete for quasi-reduction between an observable specification $M$ and complete observable implementations $I \in \faultDomain$,
as there exist techniques to complete $M$ to some $M'$ such that $I$ is a quasi-reduction of $M$ if and only if it is a reduction of $M'$. Such completions are described in~\cite{DBLP:journals/tse/Hierons17}.

\subsection{Significance of r(0)-distinguishability}\label{sec:r0impact}

As discussed in the previous section, function \textsc{GenerateTestSuite} could also be implemented using r-distinguishability as defined for quasi-reduction, considering only defined inputs. However, stronger definitions of r-distinguishability can result in fewer pairs of states being r-distinguishable. This can in turn reduce the size of maximal pairwise r-distinguishable sets and hence delay the termination of sequences in the traversal sets $Tr(s,m)$.

To provide some intuition on the impact of the definition of r-distinguishability on the generated test suite, we have computed test suites for $M_{ex}$ and the card reader example $\crd$, using three different variants of r-distinguishability:
\begin{itemize}
    \item[\texttt{rd1}] r-distinguishability as described in Definition~\ref{def:rdis} 
    \item[\texttt{rd2}] r-distinguishability as described in Definition~\ref{def:rdis}, but excluding r(0)-distinguishability (still allowing $s_1$ and $s_2$ to be r(1)-distinguished by some input $x$ defined in only one of them)   
    \item[\texttt{rd3}] r-distinguishability for (quasi-)reduction (considering only input sequences defined in both $s_1$ and $s_2$)     
\end{itemize}

\begin{table}
    \footnotesize
    \caption{Test suite sizes depending on the definition of r-distinguishability}
    \begin{center}
        \begin{tabular}{|l|l|l|l|l|}\hline\hline
            {\bf variant} & \multicolumn{2}{c|}{$M_{ex}$} & \multicolumn{2}{c|}{$\crd$} \\
            & \textbf{test cases} & \textbf{inputs} & \textbf{test cases} & \textbf{inputs} \\\hline\hline
            \texttt{rd1} & 20 & 116 & 473 & 3186 \\
            \texttt{rd2} & 25 & 146 & 1509 & 9984 \\
            \texttt{rd3} & 54 & 365 & out of memory & - \\                        
            \hline\hline
        \end{tabular}
    \end{center}
    \label{table:rd_impact}
    \normalsize
\end{table}%
Table~\ref{table:rd_impact} describes the number of test cases and the total number of inputs applied over all test cases for these variants for $M_{ex}$ and the card reader specification $\crd$, assuming in both cases that implementations have at most as many states as the specification (i.e., that $m=n$). The increase in test suite size between \texttt{rd1} and \texttt{rd2} can be attributed to \texttt{rd2} always requiring the application of at least one input to r-distinguish states. Variations \texttt{rd1} and \texttt{rd2} do not, however, differ in the pairs of states considered r-distinguishable. This is in contrast to \texttt{rd3}, which no longer allows states to be r-distinguished due to differing defined inputs. Thus, \texttt{rd3} can induce smaller maximal sets of pairwise r-distinguishable states. For $\crd$, \texttt{rd3} induces 6 such sets:
\begin{align*}
    S'_1 &:= \{\texttt{init}, \texttt{card0}\} \\
    S'_2 &:= \{\texttt{init}, \texttt{card1}\} \\
    S'_3 &:= \{\texttt{init}, \texttt{auth0}\} \\
    S'_4 &:= \{\texttt{init}, \texttt{auth1}\} \\
    S'_5 &:= \{\texttt{init}, \texttt{PIN0}, \texttt{PIN1}, \texttt{PIN2}\} \\
    S'_6 &:= \{\texttt{init}, \texttt{ejected0}, \texttt{ejected1}\}
\end{align*}
These are much smaller than the sets $S_{00}$ to $S_{11}$ computed for \texttt{rd1} in Section~\ref{sec:rdis}, each of which contained 8 states. Recall, that in $\crd$ all states are d-reachable. Thus, when using \texttt{rd3},  termination of traces in, for example, $Tr(\texttt{init},m) = Tr(\texttt{init},10)$ occurs only after visiting states of $S'_i$ exactly $m-2+1=9$ times for $1 \leq i \leq 4$ or states of $S'_5$ or $S'_6$ exactly $m-3+1=8$ times. This requires many more visits compared to using \texttt{rd1} or \texttt{rd2}, where termination of traces in $Tr(\texttt{init},10)$ occurs after visiting states of any of the four maximal pairwise r-distinguishable sets exactly $m-8+1=3$ times. Furthermore, each $S'_i$ is a proper subset of at least one of the sets $S_{00}$ to $S_{11}$.  Therefore, traces in $Tr(\texttt{init},10)$ are much longer when using \texttt{rd3}, than they are for \texttt{rd1} or \texttt{rd2}. For example, using \texttt{rd3} requires the cycle \texttt{init, card0, auth0, PIN0, ejected0, init} to be repeated 4 times, for a total of 20 transitions, until termination occurs by reaching states of $S_6$ exactly 8 times. In contrast, the same cycle is not completed once when using \texttt{rd1}, as the first three transitions already visit states of $S_{00}$ exactly 3 times.
Due to this effect, the test suite for $\crd$ and \texttt{rd3} is not feasible to compute in a reasonable amount of time or space\footnote{The experiments were performed using Ubuntu 18.04 on a Intel Core i7-4700MQ processor and 16GB RAM.} and is not included in Table~\ref{table:rd_impact}.

The test suites described in Table~\ref{table:rd_impact} and executable implementations of all three variants are available for download under \url{http://www.mbt-benchmarks.org}.

\section{Related Work}\label{sec:related}

The investigation of complete test suites derived from FSMs has a very long tradition, starting  with~\cite{chow:wmethod,vasilevskii1973} where the so-called W-Method for testing language equivalence against deterministic completely specified FSMs has been presented. These original generation methods were refined with the objective to reduce test suite sizes while preserving completeness. Today, the H-Method published in~\cite{DBLP:conf/forte/DorofeevaEY05}, the SPY method~\cite{simao_reducing_2012}, and the SPYH-method~\cite{Soucha2018} appear to be the most advanced test generation methods for language equivalence testing. 

At the same time, the investigation of complete test generation methods was extended to nondeterministic FSMs, where language equivalence is of lesser importance, since implementations typically show only subsets of the behaviours allowed by the reference model. This led to tests checking whether an implementation is a reduction of a (nondeterministic) reference model. Important publications in this direction are~\cite{petrenko_testing_2011,hierons_testing_2004}.

The utilisation of fault models has originally been advocated in~\cite{gotzhein_fault_1996} and (using the terms \emph{validity} and \emph{unbias} for exhaustiveness and soundness, respectively)~\cite{DBLP:conf/tapsoft/Gaudel95}. The concept of fault models has been refined further in~\cite{DBLP:series/natosec/Pretschner15}.

Originally, complete test suites were regarded as interesting theoretical research objects, but of lesser practical importance. This was due to their unmanageable size when deriving tests from FSM models representing more complex real-world control applications. This has changed since it has been shown that by constructing equivalence classes, the test suite size can be considerably reduced while still preserving the completeness property~\cite{peleska_sttt_2014,Huang2017}. 
This theory has been applied and evaluated for the test of  control systems with medium complexity in~\cite{Huebner2017}.
This experimental evaluation also shows that complete equivalence class test suites still exhibit considerable test strength, if the SUT is {\it not} contained in the fault model.

Partial FSMs have been defined and investigated to some extent quite early, initially without considering applications to testing. In~\cite[Chapter~7]{gill62}, the term \emph{input restricted machine} was used to introduce partial, deterministic FSMs. The unavailability of certain inputs in specific states was considered to be imposed by environments that are ``unable'' or ``unwilling'' to provide these inputs is these situations. The notion of quasi-equivalence has also been defined in~\cite{gill62} for the first time. Since only deterministic machines were considered, however, quasi-reduction or comparable terms have not been discussed. In~\cite{Starke72}, the notion of partial nondeterministic FSMs was already acknowledged, but only interpreted as  malformed specification models, since in presence of partial machines, {\it \dots the system can not function in a well-defined way.''}~\cite[p.~146]{Starke72}.

The practical relevance of partial FSMs was recognised in the 1990ies in the context of protocol specifications and protocol testing~\cite{DBLP:conf/pts/Petrenko91}. As reviewed in~\cite{DBLP:journals/tc/PetrenkoY05}, different suggestions have been made to treat missing inputs in a given state: (1) partial FSMs can been \emph{completed} by adding a self-loop for each missing input, together with a ``null-output'' indicating that the FSM does not provide any reaction to such an input~\cite{DBLP:journals/tse/SidhuL89,DBLP:journals/jcss/YannakakisL95}. This variant corresponds to ignored inputs in our classification. (2) Alternatively, a partial FSM can be completed 
by adding an \emph{error state} and creating a transition for every state and unspecified input to the error state, accompanied by an output indicating the occurrence of an error. The error state responds with a self-loop and error output to any input~\cite{DBLP:journals/jcss/YannakakisL95}. 

These techniques to construct complete FSMs from partial ones were regarded as useful in the context of protocol testing for the purpose of \emph{strong conformance testing}~\cite{DBLP:journals/jcss/YannakakisL95}, where reference models are considered to determine {\it every} behaviour of a protocol implementation. It has been criticised in~\cite{DBLP:conf/pts/Petrenko91}, however,  that the completion method is useless in a situation where the operational environment prevents the occurrence of unspecified events. This situation may be formally captured by modelling both the expected behaviour of the SUT and the operational environment as state machines and building the resulting product machine, which is only partially defined as long as the environment is not allowed to produce any input in any operation situation.

The notion of quasi-equivalence has also been discussed in~\cite{DBLP:journals/jcss/YannakakisL95}, where the definition of \emph{weak conformance} has been introduced for partial deterministic FSMs: in weak conformance, the implementation machine may exhibit any behaviour in case of an unspecified input.

None of the publications referenced above discuss the possibility that the target system itself disables the occurrence of inputs. Consequently, the practicability and the advantages of a grey-box testing approach to systems with state-dependent disabled inputs,
as well as the notion of strong reduction studied in this article have not been investigated either.

In UML and SysML~\cite{uml_2_5,SysML19}, the FSM concept of inputs triggering a transition has been generalised to \emph{triggers} denoting that the transition will accept   the occurrence of an (atomic or parameterised) signal, a change event indicating that the valuation of a specified condition changes from false to true, or an operation invocation.  
The occurrence of such an event in a state machine state where none of the transitions have a corresponding trigger leads to the event being lost for this state machine, unless the event is \emph{deferred} to be processed in a state to be visited later. This corresponds to the concept of ignored events discussed in this paper. There are no analogies in UML/SysML to the concepts of undefined events and disabled events.

It is interesting to note that in the field of labelled   transition systems, the well-known ioco-conformance relation~\cite{DBLP:conf/fortest/Tretmans08} is very similar to quasi-reduction for FSMs: with ioco, the implementation is allowed to exhibit any behaviour on sequences of events that are not contained in the so-called suspension traces of the reference transition system. 

In this article, we have considered several interpretations of input events that do not occur is a certain FSM state, but we did {\it not} consider the possibility that the environment could be {\it blocked} when offering an unspecified input to an FSM. The reason for this omission  is that the simple semantics of FSMs does not consider concepts like `blocking' or `deadlock'. These notions have been investigated in depth in the field of process algebras, in particular {\it Communicating Sequential Process (CSP)}~\cite{Hoare:1985:CSP:3921,Roscoe:1997:TPC:550448}, where synchronous channel communications can be nondeterministically refused, and communicating processes may offer or accept communications simultaneously on several channels. It is interesting to note that, despite their more expressive semantics,
 complete testing theories also exist for these algebras as shown, for example, in~\cite{DBLP:conf/pts/CavalcantiS17,PELESKA20191}.

\section{Conclusion}\label{sec:comc}

We have presented the new conformance relation \emph{strong reduction} for testing implementations against partial, nondeterministic, finite state machines. This relation is especially well-suited for the verification of systems whose inputs may be disabled or enabled, depending on the internal system state. This occurs frequently
in the case of graphical user interfaces or systems with interfaces that are mechanically enabled or disabled during system execution. It has been explained how the  new relation complements the existing FSM-related conformance relations language equivalence, reduction, quasi-equivalence, and quasi-reduction.

A new test generation algorithm producing complete suites for verifying strong reduction conformance has been introduced, and its completeness property has been proven. The tests are executed in a grey-box setting, where the state-dependent enabled and disabled inputs of the implementation are revealed in each test step. This grey-box information leads to significantly smaller test suites: complexity calculations showed that in the best case, the test suite size depends on the reference model size only in a linear way, while the known black-box algorithms for other conformance relations produce suites growing at least quadratically with the size of the reference model.

The test generator is available in the open source library {\it fsmlib-cpp}.

\appendix
\section{A 4-complete Test Suite for the Example FSM}\label{appendix:mex_test_suite}

The following 20 test cases, containing a total of 116 inputs, constitute a possible result of applying \textsc{GenerateTestSuite($M_{ex}$,$4$)}. Recall that $M_{ex}$ has been defined in Section~\ref{sec:dreach}.
 
\begin{verbatim}
a.b.a.a.a.a.b
a.b.a.a.b.a.b
a.b.a.b.a.a.b
a.b.a.b.b.a.b
a.b.b.a.a.a.b
a.b.b.a.b.a.b
a.b.b.b.a.a.b
a.b.b.b.b.a.b
a.a.a.a.b
a.a.a.b.b
a.a.b.a.b
a.a.b.b.b
b.a.a.a.b
b.a.a.b.b
b.a.b.a.b
b.a.b.b.b
b.b.a.a.b
b.b.a.b.b
b.b.b.a.b
b.b.b.b.b
\end{verbatim}

This test suite can be obtained using the \emph{fsmlib-cpp} library\footnote{Publicly available for download at \url{https://github.com/agbs-uni-bremen/fsmlib-cpp}.}, an open source project programmed in C++. 
The library contains fundamental
algorithms for processing   Mealy Machine FSMs  
and a variety of model-based test generation algorithms, including an implementation of \textsc{GenerateTestSuite} as described in Fig.~\ref{alg:GenerateTestSuite}.
Download and installation of the library are explained in the 
lecture notes~\cite[Appendix B]{PeleskaHuangLectureNotesMBT} which are also publicly available.

After installation, the \texttt{fsm-generator} executable can be employed  as follows:
\begin{verbatim}
fsm-generator -sr -a 0 <path to specification FSM>             
\end{verbatim}
Here \texttt{-sr} indicates that a test suite for strong reduction is to be computed. Next, \texttt{-a 0} specifies the maximum number of additional states of FSMs in the fault domain compared to the specification FSM $M$. Thus, in the above command, $m$ is set to $|M|+0$. Finally, a path to a file specifying $M$ is required. The expected format of such files and further available parameters of \texttt{fsm-generator} are described in~\cite[Appendix B]{PeleskaHuangLectureNotesMBT}.

\section{Proof of Complexity Results}\label{sec:complexityproofs}

In order to calculate bounds on the number of test cases to be generated in specific situations, it is useful to 
``unroll'' the test generation algorithm specified in Fig.~\ref{alg:GenerateTestSuite} into the representation shown in Fig.~\ref{alg:GenerateTestSuiteX}. There, we use the notation $\Pref_1(\bar x/\bar y)$ to denote the  set of {\it non-empty} prefixes of $\bar x/\bar y$. The cardinality is obviously 
$|\Pref_1(\bar x/\bar y)| = |\bar x|$. 
It is straightforward to see that the two algorithm representations are semantically equivalent. Throughout this section, we use notation $a = m-n$ to denote the maximal number of additional states that may be used in the SUT.

\begin{figure}
    \begin{algorithmic}[1]
        \Function{GenerateTestSuite}{$M,m$}  : $\mathbb{P}(\Sigma_I^*)$
        \State{choose a state cover $V$ of $M$}
        \State{$V' \gets \{ \bar{x}/\bar{y} \in \lang(M)~|~\bar{x} \in V \}$}
        \State{calculate $Tr(s,m)$ for each $s\in \widehat{S}$}
        \State{$T\gets \bigcup_{s \in \widehat{S}} \  \{\bar{v}_s\}.Tr(s,m)$}
        \State $D\gets \{ (s,\bar{x}/\bar{y})~|~s \in \widehat{S}, \bar{x}/\bar{y} \in \lang_M(s), \ term(s,\bar{x}/\bar{y},m) \not= \varnothing  \}$ 
        \ForAll{$\bar{x}_1/\bar{y}_1, \bar{x}_2/\bar{y}_2 \in V'$}
                \State{$s_1 \gets \ii s\after \bar{x}_1/\bar{y}_1$}
                \State{$s_2 \gets \ii s\after \bar{x}_2/\bar{y}_2$}
                \If{$s_1, s_2$ are r-distinguishable}
                    \State{$W' \gets \{ \bar{x}'~|~\{\bar{x}_1.\bar{x}',\bar{x}_2.\bar{x}'\} \subseteq \Pref(T) \}$}
                    \If{$W'$ does not r-distinguish $s_1, s_2$}
                        \State{choose $W$ that r-distinguishes $s_1, s_2$}
                        \State{$T \gets T \cup \{\bar{x}_1, \bar{x}_2\}.W$}
                    \EndIf
                \EndIf
            \EndFor
        \ForAll{$(s,\bar{x}/\bar{y}) \in D$}
            \State{choose an $S_i \in term(s,\bar{x}/\bar{y},m)$}
            \ForAll{$\bar{x}_1/\bar{y}_1\in V', \bar{x}_2/\bar{y}_2\in \{\bar{v}_s\}.\Pref_1(\bar{x}/\bar{y})$}
                \State{$s_1 \gets \ii s\after \bar{x}_1/\bar{y}_1$}
                \State{$s_2 \gets \ii s\after \bar{x}_2/\bar{y}_2$}
                \If{$s_1 \not= s_2 \wedge s_1 \in S_i \wedge s_2 \in S_i$}
                    \State{$W' \gets \{ \bar{x}'~|~\{\bar{x}_1.\bar{x}',\bar{x}_2.\bar{x}'\} \subseteq \Pref(T) \}$}
                    \If{$W'$ does not r-distinguish $s_1, s_2$}
                        \State{choose $W$ that r-distinguishes $s_1, s_2$}
                        \State{$T \gets T \cup \{\bar{x}_1, \bar{x}_2\}.W$}
                    \EndIf
                \EndIf
            \EndFor
        \EndFor    
        \ForAll{$(s,\bar{x}/\bar{y}) \in D$}
            \State{choose an $S_i \in term(s,\bar{x}/\bar{y},m)$}
            \ForAll{$(\bar{x}_1/\bar{y}_1, \bar{x}_2/\bar{y}_2) \in (\{\bar{v}_s\}.\Pref_1(\bar{x}/\bar{y}))^2$}
                \State{$s_1 \gets \ii s\after \bar{x}_1/\bar{y}_1$}
                \State{$s_2 \gets \ii s\after \bar{x}_2/\bar{y}_2$}
                \If{$s_1 \not= s_2 \wedge s_1 \in S_i \wedge s_2 \in S_i$}
                    \State{$W' \gets \{ \bar{x}'~|~\{\bar{x}_1.\bar{x}',\bar{x}_2.\bar{x}'\} \subseteq \Pref(T) \}$}
                    \If{$W'$ does not r-distinguish $s_1, s_2$}
                        \State{choose $W$ that r-distinguishes $s_1, s_2$}
                        \State{$T \gets T \cup \{\bar{x}_1, \bar{x}_2\}.W$}
                    \EndIf
                \EndIf
            \EndFor
        \EndFor
        \State{$T \gets \{ \bar{x} \in T~|~\nexists \bar{x}' : \bar{x}' \not= \epsilon \wedge \bar{x}\bar{x}' \in T  \} $}
        \State \Return $T$
        \EndFunction
    \end{algorithmic}
    \caption{Algorithm generating $m$-complete $\strongred$-conformance test suites -- unrolled version equivalent to the algorithm in Fig.~\ref{alg:GenerateTestSuite}.}
    \label{alg:GenerateTestSuiteX}
\end{figure}

\subsection{Deterministic, completely specified case.}\label{sec:appdeter}
If the reference model $M$ is deterministic and completely specified, $M$ can also be assumed to be minimised. Under these assumptions, all states are d-reachable, so $|V| = |S| = n$, and $\widehat S = S$. Moreover, all states are 
pairwise distinguishable, so $S_D = \{ S \}$. Therefore, every step through the FSM visits a state of the
only element $S$ of $S_D$. The termination criterion to reach $m -|\widehat{S_i}| + 1$ states of some $S_i\in S_D$
is simplified to $m -|\widehat{S_i}| + 1 = m - |\widehat S| + 1 = m - |S| + 1 = m-n+1 = a+1$.
Thus, the definition of the sets $Tr(s,m)$ can be re-written as 
\begin{align*}
Tr(s,m) & = \Pref \{ \bar{x}~|~\exists \bar{y} : \  \bar{x}/\bar{y} \in \lang_M(s) \wedge \ term(s,\bar{x}/\bar{y},m) \not= \varnothing \} \\
& = \Pref \{ \bar{x}~|~|\bar{x}| = a + 1 \} \\
& = \bigcup_{i=0}^{a+1} \Sigma_I^i
\end{align*}
Note that for the case considered here, the sets $Tr(s,m)$ are in fact independent of the state $s$.

Therefore, the set $T$ is initialised to $T_\text{init} \leftarrow V.\bigcup_{i=0}^{a+1} \Sigma_I^i$ in line~5 of the algorithm. In the loop of the algorithm from line~7 to line~20, this set may be extended, but the elements initially inserted remain unchanged. In line~21, the test case set $T$ is finally reduced to input sequences of maximal length. From the initial assignment in line~5, just the elements from $V.\Sigma_I^{a+1}$ remain. The cardinality of this set has the following upper bound.
\begin{equation}\label{eq:tinit}
|T_\text{init}| = |V.\Sigma_I^{a+1}| \le |V|\cdot |\Sigma_I^{a+1}| = n \cdot |\Sigma_I|^{a+1}
\end{equation} 

The set $D$ assigned in line~6 of the algorithm is calculated for the case considered here by
\begin{align*}
D & = \{ (s,\bar{x}/\bar{y})~|~s \in \widehat{S}, \bar{x}/\bar{y} \in \lang_M(s), \ term(s,\bar{x}/\bar{y},m) \not= \varnothing  \} \\
& =    \{ (s,\bar{x}/\bar{y})~|~s \in S, \bar{x}/\bar{y} \in \lang_M(s),  |\bar x| = a+1  \} \\
& = \{ (s,\bar{x}/\bar{y})~|~s \in S, \bar x\in\Sigma_I^{a+1}, \bar{x}/\bar{y} \in \lang_M(s)  \}
\end{align*}
Since $M$ is deterministic and completely specified, there exists exactly one output sequence $\bar y$ for any given input $\bar x\in \Sigma_I^{a+1}$. Therefore, the cardinality of $D$ is
\begin{equation}\label{eq:cardD}
|D| = |S| \cdot |\Sigma_I|^{a+1} = n \cdot |\Sigma_I|^{a+1}
\end{equation}
and corresponds to the number of cycles performed by the outer loops in lines~18 and 32 of the algorithm in Fig.~\ref{alg:GenerateTestSuiteX}.

We will now elaborate bounds for the number of input sequences added to $T$ in each of the loops starting in lines~7, 18, and 32.

\subsubsection*{Loop in line~7.} Since $M$ is deterministic and completely specified by assumption, each input sequence stimulates exactly one output sequence. Therefore, $|V'| = |V| = n$. The number of distinct pairs 
$\bar{x}_1/\bar{y}_1 \neq \bar{x}_2/\bar{y}_2 \in V'$ is $\binom{n}{2} = \frac{1}{2}(n^2 - n)$. For a worst-case estimate, it is assumed that the if-condition in line~10 always evaluates to true, so that $T$ is extended by exactly two input sequences in line~14 (note that $W$ contains just one element distinguishing $s_1, s_2$). Summarising, the loop in line~7 adds at most $n^2 - n$ new elements to $T$.

\subsubsection*{Loop in line~18.} The outer loop in line~18 is executed $|D|$ times. Therefore, the upper bound of elements to be added in the inner loop starting in line~20 needs to be multiplied by $n \cdot |\Sigma_I|^{a+1}$ according to Equation~\eqref{eq:cardD}. 

The inner for-loop starting in line~20 is executed $|V'|\cdot |\bar x| = n \cdot (a+1)$ times.
In line~19, $S$ is always chosen, since $S_D = \{ S \}$. Therefore, the if-condition in line~23 evaluates to true whenever $s_1\neq s_2$. For the worst-case upper bound, it is assumed that this is always true and that $W'$ defined in line~24 never distinguishes $s_1$ and $s_2$. Thus, the inner loop adds at most $2\cdot n \cdot (a+1)$
elements to $T$. 

Multiplied with number of iterations of the outer loop starting in line~18, this results in an upper bound 
of $$n \cdot |\Sigma_I|^{a+1}\cdot 2\cdot n \cdot (a+1) = 2\cdot n^2\cdot(a+1)\cdot |\Sigma_I|^{a+1}$$
elements to be added to $T$.

\subsubsection*{Loop in line~32.} Again, the outer loop in line~32 is executed $|D| = n \cdot |\Sigma_I|^{a+1}$ times. The if-block in lines~39---42 of the inner loop starting in line~34 is executed at most
$\binom{|\{\bar{v}_s\}.\Pref_1(\bar{x}/\bar{y})|}{2}$ times. Recalling that 
$|\Pref_1(\bar{x}/\bar{y})| = a + 1$, that $\binom{a+1}{2} = \frac{1}{2}\cdot(a^2+a)$, and that   two input sequences are added in line~41, this leads to an
upper bound of 
$$
n \cdot |\Sigma_I|^{a+1} \cdot (a^2+ a) 
$$
of elements to be added to $T$ when executing the outer loop starting in line~32.

\subsubsection*{Overall estimate.}
We do not have an estimate for the prefixes to be deleted from $T$ in line~46, originating from input sequences previously added during execution of the three loops. Therefore, we use the sum of $|T_\text{init}|$ after
removal of prefixes plus
the maximal contribution of elements in the three loops  as upper bound $B$, that is,
\begin{align}
B = &\big(n \cdot |\Sigma_I|^{a+1}\big) + \big( n^2 - n \big) + \big(  2\cdot n^2\cdot(a+1)\cdot |\Sigma_I|^{a+1} \big) + \big( n \cdot |\Sigma_I|^{a+1} \cdot (a^2+ a) \big) \label{eq:deterbound} \\
= & n^2\cdot|\Sigma_I|^{a+1} \cdot 2(a+1) + n^2 +  n\cdot|\Sigma_I|^{a+1}\cdot(a^2 + a + 1) - n  \nonumber
\end{align}


\subsection{Best-Case Test Suite Size Reduction Effect from Grey-Box Testing}\label{sec:bestcase}

Next, we consider reference FSMs $M$ with the following properties
\begin{enumerate}
\item $M$ is nondeterministic and partial,  so that {\it all} states are pairwise r(0)-distinguishable (i.e.~$\Delta_M(s_1) \neq \Delta_M(s_2)$ for all $s_1\neq s_2\in S$).
\item All states of $M$ are d-reachable.
\end{enumerate}
These properties represent the edge case of a nondeterministic, partial model which is {\it best-suited} for reduction testing, since complete test suites can be created with a minimal amount of cases.

Since all states are d-reachable, $V = S$. Since all states are r(0)-distinguishable, $\widehat S = S$ and
$S_D = \{ S \}$. The termination condition to reach states of some $S_i\in S_D$ at  $m - |\widehat{S_i}| +1$ times is simplified as $m - n + 1 = a + 1$. Therefore, the sets $Tr(s,m)$ are again independent of $s$ and have the value
$$
Tr(s,m) = \bigcup_{i=0}^{a+1}\Sigma_I^i
$$
just as in the deterministic, completely specified case handled above. As a consequence, the initialisation of
set $T$ in line~5 of the algorithm has the same result as in the previous case, so its initial cardinality of $T$ after removal of prefixes is bounded  again by 
$$
|T_\text{init}| \le n\cdot|\Sigma_I|^{a+1}
$$

So far, the same properties hold as in the deterministic case investigated above.
The set $D$ is initialised  again in line~6 to 
$$
D = \{ (s,\bar{x}/\bar{y})~|~s \in S, \bar x\in\Sigma_I^{a+1}, \bar{x}/\bar{y} \in \lang_M(s)  \}
$$
The cardinality of $D$, however, is generally higher than in the deterministic case, because more than one $\bar x/\bar y\in L_M(s)$ can exist for a given input sequence $\bar x$ if $M$ is nondeterministic. It will be shown in the subsequent analysis that this does not affect the upper bound of test cases to be added to the suite.

The cardinality of $V'$ may similarly be greater than that of $V$, since one input sequence from $V$ may lead to different outputs, due to nondeterminism. By assumption, however, all states are d-reachable. Therefore, for a fixed $\bar v\in V$,  all pairs $\bar v/\bar y\in V'$ reach the same target state $\bar v_s$.

As before, we will now elaborate bounds for the number of input sequences added to $T$ in each of the loops starting in lines~7, 18, and 32.

\subsubsection*{Loop in line~7.} 
Suppose that the states $s_1, s_2$ assigned in lines~8 and 9 of this loop are different, and, therefore, r(0)-distinguishable by assumption. Thus, an assignment in line~14, only adds
the pair $\{ \bar{x}_1, \bar{x}_2 \}$  to $T$, since $W$ is the empty set. Since $\bar{x}_1/\bar{y}_1, \bar{x}_2/\bar{y}_2 \in V'$, however, these input sequences $\bar{x_1}, \bar{x_2}$ are already contained in the state cover
$V$, and, therefore, also contained as prefixes in the set $T$ initialised in line~5.

Summarising, this loop does not add any new element to $T$ for the case considered here.

\subsubsection*{Loop in line~18.}
As explained for the loop in line~7, the sets $W$ used in the assignment to $T$ in line~27 are always empty, since all pairs of non-equal states are r(0)-distinguishable. Moreover, every input sequence $\bar x_1$ is already contained in $T$ and, therefore, does not extend this set. Now consider the input sequences $\bar x_2$ arising from traces $\bar{x}_2/\bar{y}_2\in \{\bar{v}_s\}.\Pref_1(\bar{x}/\bar{y})$. By construction of $D$, the input sequence $\bar x_2$ is always a non-empty prefix of $\bar x$, and, therefore, 
$\bar x_2\in\bigcup_{i=1}^{a+1}\Sigma_I^i$. Consequently, $\bar x_2$ is also already contained as some prefix of
an input sequence in the set $T$ as initialised in line~5.

Summarising, this loop also does not add any new element to $T$ for the case considered here.

\subsubsection*{Loop in line~32.}
The same argument as given for the previous loop yields the fact that no new input sequences are add to $T$ in the third loop.

\subsubsection*{Overall estimate.}
As a consequence of none of the loops adding further test cases to $T_\text{init}$, an upper bound for the number of test cases to be produced in this case is given by a bound for the cardinality of $T_\text{init}$, which  is $n\cdot|\Sigma_I|^{a+1}$.

\subsection{Worst Case}\label{sec:worstcase}

The worst case size of complete test suites occurs if 
\begin{itemize}
\item the reference model is nondeterministic,
\item the initial state is the only d-reachable one, and
\item no pair of states is reliably distinguishable (so we can consider the reference model $M$ to be completely specified). 
\end{itemize}
For this case, we have $V = \{ \epsilon \}$, $V' = \{ \epsilon \}$, $\widehat S = \{ \ii s \}$, and 
$S_D = \{ \{ s \}~|~s\in S\}$ (recall that $S_D$ contains singleton sets if no pair of states is reliably distinguishable). The termination criterion $term(s,\bar x/\bar y,m) \neq \varnothing$ requires in this case, that in applying $\bar x/\bar y$ to $s$
\begin{itemize}
\item either the initial state of $M$ is visited $m - |\widehat{\{ \ii s \}}| + 1 = m - |\{ \ii s \}| + 1 = m$ times, 
\item or any other state $s \neq \ii s$ is visited $m - |\widehat{\{ s \}}| + 1 = m - |\varnothing| + 1 = m+1$  times.
\end{itemize}
Consider now the longest possible trace $\bar{x}/\bar{y} \in \lang_M(s)$ that is not terminated.
By the above criterion, this sequence reaches $\ii s$ at most $(m-1)$ times and any of the $(n-1)$ other states of $M$ at most $m$ times, while any extension of it by a single input must be terminated. Therefore, $(m-1) + (n-1) \cdot m = mn - 1$ constitutes an upper bound on the length of traces not terminated. The tightness of this upper bound is demonstrated by FSMs in which from each state $s_i$ there exists only a single transition, forming a single cycle $\ii{s},s_1,\ldots,s_{n-1},\ii s$ of $n$ states. In such FSMs, a trace of length $mn-1$ applied to the initial state visits each state $(m-1)$ times, whereas a trace of length $mn$ visits the initial state $m$ times and thus terminates.

From this upper bound it follows that for any trace $\bar{x}/\bar{y} \in \lang_M(s)$  of length $mn$, one of the prefixes of $\bar{x}/\bar{y}$ must be terminated. Hence, as $M$ is effectively completely specified, the sets $Tr(s,m)$ are each a subset of the set of all input sequences of length up to $mn$
$$
Tr(s,m) \subseteq \bigcup_{i=0}^{mn} \Sigma_I^i
$$
Since $\widehat S = \{ \ii s \}$ and $V = \{ \epsilon \}$, the initialisation of $T$ in line~5 of the algorithm results in
$$
   T = \{ \epsilon \}.Tr(\ii s,m) \subseteq \bigcup_{i=0}^{mn} \Sigma_I^i
$$
Removing the true prefixes from input sequences in $T$ results in a set of cardinality
$$
|T_\text{init}| \leq |\Sigma_I^{mn}| = |\Sigma_I|^{mn}	
$$ 

The three loops do not extend $T$, because this only happens for states $s_1, s_2$ that are r-distinguishable, which is never the case for the model $M$ considered here. Summarising, the total number of test cases to be executed in a complete test suite is at most $|\Sigma_I|^{mn}$.

\ack The authors would like to thank Wen-ling Huang and Rob Hierons for helpful comments on initial versions of this article.

\bibliographystyle{wileyj}
\bibliography{references,jp}

\end{document}